\documentclass[10pt,times,amsfonts,twocolumn,amsmath,amssymb,superscriptaddress,aps,prx]{revtex4-2}

\usepackage{tabularx} 
\usepackage{graphicx}

\usepackage[caption=false]{subfig}

\usepackage{color}
\usepackage{comment}
\usepackage{color}
\usepackage[colorlinks,urlcolor=blue,citecolor=blue,linkcolor=blue]{hyperref}
\usepackage{cleveref}
\usepackage{physics}
\usepackage{mathtools}
\usepackage{mathrsfs}
\usepackage{booktabs}
\usepackage{bm}
\usepackage{amsthm}
\usepackage[english]{babel}
\newtheorem{theorem}{Theorem}

\newtheorem{lemma}{Lemma}
\newtheorem{definition}{Definition}
\usepackage{amsmath}
\newtheorem{proposition}{Proposition}
\newtheorem{corollary}{Corollary}

\usepackage{needspace}
\usepackage{physics}
\usepackage{bm}
\usepackage[page,toc,titletoc,title]{appendix}
\usepackage{fixmath}
\usepackage{dsfont}
\hypersetup{
    colorlinks,
    citecolor=blue,
    filecolor=black,
    linkcolor=blue,
    urlcolor=blue
}
\usepackage{graphicx}
\usepackage{graphicx,accents}

\usepackage{mathtools}

\begin{document}

\title{Distinguishing quantum processes with bounded coherent memory}

\author{Magdalini Zonnios}
\email{zonniosm@tcd.ie}
\affiliation{School of Physics, Trinity College Dublin, College Green, Dublin 2, D02 K8N4, Ireland}
\affiliation{Trinity Quantum Alliance, Unit 16, Trinity Technology and Enterprise Centre, Pearse Street, Dublin 2, D02 YN67, Ireland}

\author{Felix C. Binder}
\email{felix.binder@tcd.ie}
\affiliation{School of Physics, Trinity College Dublin, College Green, Dublin 2, D02 K8N4, Ireland}
\affiliation{Trinity Quantum Alliance, Unit 16, Trinity Technology and Enterprise Centre, Pearse Street, Dublin 2, D02 YN67, Ireland}

\date{\today}

\begin{abstract}
Distinguishing multi-time quantum processes is a fundamental task underlying the diagnosis, benchmarking, and learning of temporally correlated quantum dynamics. The standard benchmark for distinguishing two processes is the strategy-norm distance, which optimizes over arbitrary adaptive probing strategies but can require large coherent memory and time-dependent control. We introduce machines for autonomous distinction~($\mathsf{MAD}$s): probing strategies that apply the same quantum instrument at each time step, retain the full classical outcome record, and carry a coherent memory of dimension $d_A$. Optimizing over these strategies defines a memory-parametrized distinguishability measure, $d^{(N)}_{\mathsf{MAD}}(\mathbf{P}^N,\mathbf{Q}^N;d_A)$. We show that the resulting hierarchy is monotone in coherent memory and complete at finite times. Specifically, any admissible $N$-step probing strategy can be compiled into a single $\mathsf{MAD}$ with an internal counter and sufficiently large coherent memory, so the hierarchy saturates the strategy-norm benchmark. For recurrent processes generated by repeated system--environment interactions, we derive a single-step description that separates the generation of new distinguishing information from the propagation and decay of information generated at earlier times. Numerical results in a repeated-interaction model show that increasing coherent memory systematically improves the $\mathsf{MAD}$ success probability and closes the gap to the strategy-norm distance while remaining substantially more tractable to evaluate. $\mathsf{MAD}$ distinguishability therefore provides an operational and scalable framework for quantifying what can be learned about genuinely multi-time quantum processes with bounded coherent memory.
\end{abstract}

\maketitle

\section{Introduction}

Quantum processes describe how systems evolve in time and how they respond to interventions performed at different moments. They are naturally represented within the quantum-comb formalism~\cite{chiribella2009theoretical,PollockRodriguezRosarioFrauenheimPaternostroModi2018,MilzPRX2021Qprocesses}. A natural operational question is whether two such processes can be told apart by an experimenter who is allowed to interact with the system as it evolves. In the single-time step setting, this reduces to the familiar problem of distinguishing quantum states or channels~\cite{watrous2018theory,GilchristLangfordNielsen2005}. 
For multi-time processes, the task is richer: the experimenter may probe the system repeatedly, store information from earlier interactions, and use it to inform later ones (as depicted in the top panel of Fig.~\ref{fig:process_tester_intro}). The distinguishability of two processes therefore depends on the temporal correlations that connect different times, as well as on the capabilities of the probing strategy~\cite{GutoskiWatrous2007,gutoski2012measure,ZambonDistinguishability2024}.

\begin{figure}[t!]
    \centering
    \includegraphics[width=0.92\linewidth]{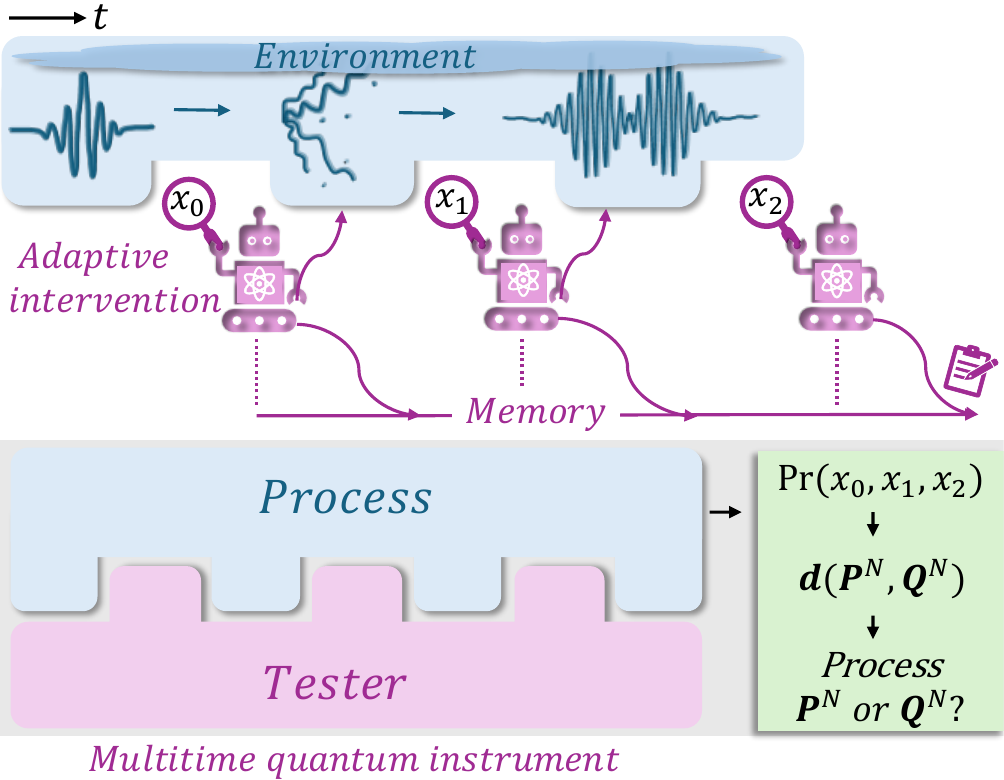}
   \caption{
Multi-time process discrimination as an operational task. 
An experimenter probes a process at several times, stores information in memory, and obtains an outcome record $(x_0,x_1,x_2)$. 
Equivalently, a tester, i.e. a multi-time quantum instrument, is inserted into the slots of a process comb to produce statistics used to decide whether the process was $\mathbf{P}^N$ or $\mathbf{Q}^N$. 
}
\label{fig:process_tester_intro}
\end{figure}

This dependence on the probing strategy is especially important when the experimenter has limited memory. In a multi-time discrimination experiment, the probing strategy, or \textit{tester}, consists of interventions inserted into the process slots, together with the memory and final decision rule used to process the resulting data. Such memory is operationally important because temporal correlations can only be exploited if information obtained at earlier intervention times can be carried forward and used later~\cite{Pollock2018,milz2020,taranto2024hierarchy}. A fully general tester may store this information coherently and use it adaptively in later probes~\cite{GutoskiWatrous2007,gutoski2012measure,chiribella2008memory}, whereas a more realistic tester may retain a classical outcome record but only a bounded amount of coherent quantum memory. Thus, memory acts as a resource determining which temporal correlations can be accessed~\cite{taranto2021,giarmatzi2023,ohst2026}. The central question of this work is: \textit{how well can two multi-time quantum processes be distinguished when the process discrimination strategy has limited coherent memory}?

The relevance of this question comes from the fact that coherent memory is itself a physical resource. In realistic devices, long-lived coherent memory is costly, noise-sensitive, and often the limiting component of a multi-time experiment~\cite{preskill2018nisq,terhal2015quantum,heshami2016quantum}, whereas classical outcome records can usually be retained and processed with comparatively little overhead. In this way, process distinguishability is recast as an operational question about accessible temporal information: which differences between processes can be detected by an agent with limited quantum memory \cite{chiribella2009theoretical,pollock2018nonmarkovian,MilzPRX2021Qprocesses,ZambonDistinguishability2024}.

Importantly, quantum process discrimination is a fundamental task in quantum information theory, appearing in channel and network discrimination~\cite{gutoski2012measure,wilde2020amortized,hirche2023network}, channel estimation~\cite{katariya2021geometric}, cryptographic security~\cite{portmann2022security}, and theories of dynamical resources~\cite{wang2019asymmetric,nakahira2021general}. For multi-time processes, quantum combs provide the natural framework for describing sequences of interventions on systems with temporal correlations~\cite{chiribella2009theoretical,pollock2018nonmarkovian,Milz2017,milz2021quantum,OrtegaTaberner2024Unifying}. In this setting, memory is revealed operationally through multi-time experiments~\cite{Pollock2018,milz2020,taranto2024hierarchy}, with experimental demonstrations in superconducting quantum processors~\cite{white2020demonstration,zhang2022predicting,nakamura2024gate} and recent theoretical developments involving memory kernels, long-range temporal structures, and process-tensor distinguishability~\cite{jorgensen2020discrete,dowling2024capturing,ZambonDistinguishability2024}.

The standard benchmark for distinguishing multi-time processes is the strategy-norm distance $d_{\mathrm{str}}^{(N)}$, which optimizes over all admissible adaptive distinguishing strategies~\cite{gutoski2012measure,watrous2018theory}. Operationally, $d_{\mathrm{str}}^{(N)}$ quantifies the optimal bias in deciding whether the process under test is $\mathbf{P}^N$ or $\mathbf{Q}^N$. Equivalently, this is a binary hypothesis-testing problem~\cite{helstrom1969,bae2015}. For equal priors, the optimal success probability satisfies $p_{\mathrm{succ}}^{\mathrm{str},\ast}=\frac12(1+d_{\mathrm{str}}^{(N)})$, so maximizing distinguishability is equivalent to minimizing the average type-I and type-II error probability $(\alpha+\beta)/2$~\cite{helstrom1969,bae2015}. As for the trace and diamond-norm distances, the strategy-norm distance is given by half the relevant norm of the difference, $d_{\mathrm{str}}^{(N)}=\frac12|\mathbf{P}^N-\mathbf{Q}^N|_{\mathrm{str}}$. Although fully general, this benchmark may require coherent memory across many time steps, time-dependent control, and the solution of a semidefinite program over positive semidefinite matrices of dimension $\sim d_S^{2N}$ with $\sim d_S^{4N}$ constraints, where $d_S$ is the system dimension. To identify the optimal distinguishing strategy with the state-of-the-art interior-point methods, this leads to a runtime scaling of order $\mathcal{O}(d_S^{12N})$~\cite{Jiang2020AFI}.

This motivates a memory-resolved approach to process discrimination. Rather than asking only how distinguishable two processes are in principle, we ask how distinguishable they are in practice when the available tester has bounded coherent memory and limited control. Such a framework should isolate coherent memory from classical record-keeping, provide a hierarchy that converges to the strategy norm as memory increases, and remain meaningful when full process reconstruction or strategy-norm optimization is infeasible.

\subsection*{Contributions}

We introduce \textit{machines for autonomous distinction} ($\mathsf{MAD}$s) and $\mathsf{MAD}$ testers -- process discrimination strategies that repeatedly apply a recurrent quantum instrument, retain the full classical outcome record, and carry a coherent quantum memory of dimension $d_A$. Optimizing over such strategies defines a memory-parametrized distinguishability measure, the $\mathsf{MAD}$ distinguishability $d^{(N)}_{\mathrm{MAD}}(\mathbf{P}^N,\mathbf{Q}^N;d_A)$. The central value of this framework is that it turns process distinguishability into a memory-resolved question, asking not only how different are two processes in principle, but how much of that difference is accessible to an agent with limited coherent memory? In this sense, $\mathsf{MAD}$ distinguishability measures accessible temporal information \cite{chiribella2009theoretical,pollock2018nonmarkovian,MilzPRX2021Qprocesses,ZambonDistinguishability2024}.

Our first contribution is a coherent-memory hierarchy for multi-time process discrimination. For fixed $d_A$, $\mathsf{MAD}$ distinguishability $d_{\mathsf{MAD}}$ gives the best distinguishing bias achievable by testers with that amount of coherent memory, and hence a lower bound on the strategy-norm distance. We prove that this hierarchy is monotone in $d_A$ and complete at finite times, meaning that for any fixed process length $N$, there is a finite memory dimension for which $d_{\mathsf{MAD}}$ equals the strategy-norm distance. This follows by showing that any admissible $N$-step adaptive tester can be compiled into a single $\mathsf{MAD}$ tester using an internal counter and sufficiently large coherent memory \cite{GutoskiWatrous2007,gutoski2012measure,watrous2018theory}. The hierarchy therefore does more than improve an approximation as $d_A$ increases. It reveals how much coherent memory is needed to access the temporal information that distinguishes the processes.

This makes the gap between $d^{(N)}_{\mathrm{MAD}}$ and the strategy norm operationally meaningful. If the hierarchy converges rapidly with $d_A$, then the relevant temporal correlations are compressible into a small coherent memory. If a gap persists, then the processes differ in information that a small-memory tester cannot coherently carry forward and exploit. The gap is therefore not just a limitation of the restricted tester class but also a witness of the coherent temporal memory required by the discrimination task \cite{milz2020,taranto2021,taranto2024hierarchy,ohst2026}.

This perspective separates coherent memory from other resources in a multi-time experiment. Existing approaches to restricted discrimination often constrain adaptive structure, network architecture, causal order, admissible inputs, or separability conditions \cite{chiribella2008memory,harrow2010adaptive,jencova2016conditions,bavaresco2021strict,salek2022usefulness,nakahira2021restricted,ohst2026}. By contrast, the $\mathsf{MAD}$ framework constrains coherent memory directly while retaining the full classical outcome record. The hierarchy therefore turns coherent memory into an operationally tunable resource. By varying $d_A$ while retaining the full classical record, one can isolate how much of the process distinguishability comes from coherently stored temporal information \cite{milz2020,taranto2021,giarmatzi2023,taranto2024hierarchy,ohst2026}.

Our second contribution is a recurrent description of distinguishability for processes generated by repeated system--environment interactions. We derive a single-step picture that separates the generation of new distinguishability from the propagation and decay of distinguishability generated earlier. This provides a dynamical account of how distinguishing information accumulates over time and connects naturally to collision models \cite{KretschmannWerner2005,ciccarello2022collision}, tensor-network methods for open dynamics \cite{jorgensen2020discrete,cygorek2022compression,fux2021efficient,ivander2024unified,dowling2024capturing}, and quantum reservoir or recurrent quantum models \cite{fujii2021quantum,NakajimaFujiiNegoroMitaraiKitagawa2018,mujal2021opportunities,govia2021reservoir,suzuki2022natural,lyu2026variational}.

Our third contribution is a scalable variational route to memory-resolved process discrimination. Rather than optimizing over all admissible testers, the $\mathsf{MAD}$ hierarchy optimizes over recurrent testers with increasing coherent memory dimension, yielding a controlled sequence of lower bounds on the strategy norm. In this sense, $d_A$ plays a role analogous to a bond dimension. Increasing it enlarges the expressive power of the tester while remaining tied to a physically meaningful resource \cite{Orus2014,CiracVerstraete2021}. This connects naturally to tensor-network approaches for simulating, learning, and compressing non-Markovian dynamics \cite{fux2024oqupy,cygorek2024ace,keeling2025processTensorApproaches}.

Our numerical examples support this interpretation in a repeated-interaction model. Increasing $d_A$ systematically improves the $\mathsf{MAD}$ success probability and closes the gap to the strategy norm while avoiding the full SDP over general testers. The framework is therefore particularly useful when temporal correlations can be compressed into a small coherent memory, as expected for mixing dynamics, finite-memory environments, collision models, or compact process-tensor representations \cite{KretschmannWerner2005,ciccarello2022collision,jorgensen2020discrete,cygorek2022compression,dowling2024capturing}. Conversely, slowly mixing, non-ergodic, near-critical, or otherwise highly structured environments may require substantially larger memories \cite{fux2021efficient,cygorek2022compression,ivander2024unified}.

Finally, the framework suggests applications in process-tensor learning, benchmarking, and noise characterization. Rather than reconstructing a full multi-time process, one can ask which process differences are visible to a finite-memory probing agent, a perspective that is complementary to recent reinforcement-learning approaches for quantum processes with memory~\cite{lumbreras2026reinforcement}. This is relevant for detecting temporally correlated noise, crosstalk, calibration drift, and finite-memory environmental effects under realistic experimental constraints \cite{white2020demonstration,whiteNonMarkovian2022PRX,Zhang2022NonMarkovianSuperconducting,WhitePRL2023,Tripathi2024BenchmarkingGates,Hashim2025QCVV,white2025unifyingPRX,zhang2025learningForecasting}. The recurrent structure of $\mathsf{MAD}$ testers also connects naturally to quantum reservoir computing and recurrent quantum models \cite{Ghosh2019_Reservoir_npj,fujii2021quantum,nakajima2019boosting,NakajimaFujiiNegoroMitaraiKitagawa2018,mujal2021opportunities,govia2021reservoir,suzuki2022natural,sannia2024dissipation,martinezpena2025inputDependence,wringe2025reservoirBenchmarks}, as well as (quantum) computational mechanics \cite{CrutchfieldYoung1989,ShaliziCrutchfield2001,GuEtAl2012,Thompson2017UsingQuantumTheory,binderPRLpractical2018,Distinguishability_Yang_2020_PRE,YangEtAl2025QuantumDimensionReduction,lyu2026variational}. Taken together, these results establish $\mathsf{MAD}$ distinguishability as both a tractable approximation to the strategy norm and a tool for probing the operational structure of temporal quantum processes.

The paper is organized as follows. In Sec.~\ref{sec:background}, we review quantum combs, testers, and the strategy-norm distance. In Sec.~\ref{sec:MAD-definition}, we introduce $\mathsf{MAD}$ testers and define the corresponding memory-constrained distinguishability measure. In Sec.~\ref{sec:mad_distinguishability}, we prove the structural properties of the $\mathsf{MAD}$ hierarchy, including monotonicity and finite-time completeness. In Sec.~\ref{sec:repeated_interaction_picture}, we specialize to recurrent processes and derive the single-step description of distinguishability generation and propagation. In Sec.~\ref{sec:numerics}, we present numerical examples comparing bounded-memory $\mathsf{MAD}$ strategies with the strategy-norm distance. We conclude in Sec.~\ref{sec:conclusion} with a discussion of possible extensions.

\section{Background: combs, testers, and the strategy-norm distance}
\label{sec:background}

\subsection{Quantum processes as quantum combs}
\label{sec:combs}

\begin{figure}
  \centering
  \includegraphics[width=\columnwidth]{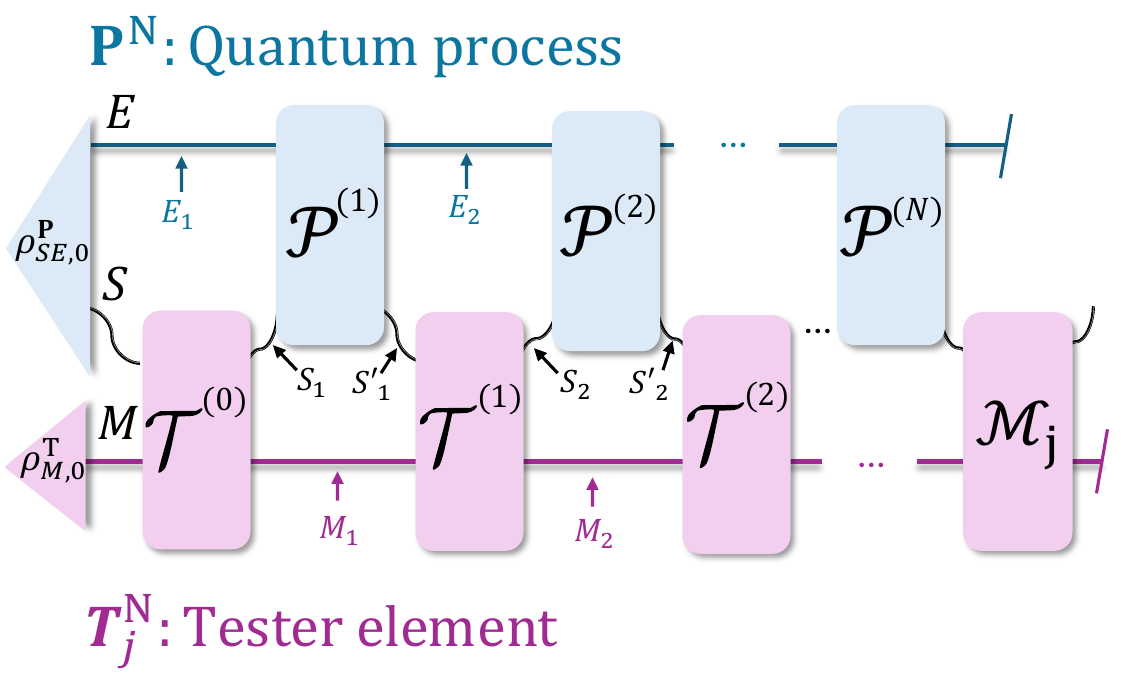}
  \caption{A quantum process over $N$ time steps can be represented as a quantum comb $\mathbf{P}^{N}$ (blue), with open ``slots'' on the system where interventions may be performed. A multi-time tester $\{\mathbf{T}^{N}_j\}_j$ (pink) is a correlated sequence of operations (potentially with memory) inserted into these slots to produce an outcome $j$. Without loss of generality, any tester can be written as a sequence of CPTP maps followed by a final quantum instrument $\{\mathcal{M}_j\}_j$, so that all intermediate measurements can be deferred to a single measurement at the end. The number of possible outcomes grows as the product of the number of outcomes at each time step, and is therefore generally exponential in $N$. Contracting a tester element with the comb yields the outcome probability $\Pr(j|\mathbf{P}^{N})$.}
  \label{fig:tester_and_comb}
\end{figure}

Multi-time quantum processes are naturally described within the process-tensor or \textit{quantum-comb} framework. We follow standard process-tensor / quantum-comb formalism and notation from
\cite{chiribella2009theoretical,pollock2018nonmarkovian,MilzPollockModi2016} which is also outlined in Appendix~\ref{app:choi-combs-testers}. An $N$-step quantum comb $\mathbf{P}^{N}$ represents the input--output behavior of an inaccessible system--environment evolution across $N$ intervention times. Operationally, $\mathbf{P}^{N}$ maps interventions on the system at different times, to combs, channels, states and outcome probabilities depending on the type of interventions applied. Any quantum comb can be realized by a sequence of completely positive trace preserving (CPTP) maps $\mathcal{P}^{(k)}\equiv\mathcal{P}^{(k:k+1)}$ which evolve the system and environment from time $k$ to $k+1$, connected via an environment,

\begin{align}\label{eqn:comb_as_sequence}
\mathbf{P}^{N}
:=
\tr_{E_{N+1}}\!\left(
\mathcal{P}^{(N)}\circ_{E_{N}}\dots\circ_{E_{2}} \mathcal{P}^{(1)}\circ_{E_1}[\rho_{SE,0}^{\mathbf{P}}]
\right),
\end{align}
where $\circ_{E_{k}}$ and $\tr_{E_{k}}$ denote the composition and partial trace over the environment space at the $k^{\text{th}}$ time step. One may equivalently think of the composition over the environment spaces as a link product of the corresponding Choi states (see Appendix~\ref{app:choi-combs-testers} for details). The state $\rho^{\mathbf{P}}_{SE,0}$ is the initial system-environment state specified by the process. We use the system, environment and memory space labeling as shown in Fig.~\ref{fig:tester_and_comb} throughout where the primed spaces $S^{'}_k$ denote the system output of the process at time $k$. Since the final environment is traced out and the maps are composed over the environment spaces, the process $\mathbf{P}^N$ specifies the effect of the environment, but is itself an operation only on the system between applications of the maps $P^{(k)}$. 
The corresponding Choi operator $\Upsilon^{\mathbf{P}^{N}}$ is positive semidefinite and satisfies multi-time causality and normalisation constraints \cite{chiribella2009theoretical}. The explicit Choi conventions and constraints are given in Appendix~\ref{app:choi-combs-testers}. To probe and extract information from such multi-time processes, one requires a
notion of multi-time measurements, which we introduce next.

\subsection{Testers as multi-time instruments}\label{sec:testers}
A \textit{tester} is the multi-time analogue of a quantum instrument. It specifies a (possibly correlated) sequence of operations that is inserted into the slots of a process to produce classical outcomes. Formally, a tester is a collection of CP (completely positive), trace-nonincreasing combs $\{\mathbf{T}^{N}_j\}_j$ whose sum $\mathbf{T}^{N}:=\sum_j \mathbf{T}^{N}_j$ is a causal and trace-preserving comb. 

Given a process comb $\mathbf{P}^{N}$, the probability of obtaining outcome $j$ is given by the contraction of $\mathbf{T}^{N}_j$ with $\mathbf{P}^{N}$,
\begin{equation}
\label{eq:born-rule-combs}
\begin{split}
\Pr(j|\mathbf{P}^{N}, \mathbf{T}^{N}) &=\tr(\mathbf{P}^{N}[\mathbf{T}^{N}_j])\\&= \tr\!\left(\Upsilon^{\mathbf{P}^{N}}\,(\Upsilon^{\mathbf{T}^{N}_j})^T\right),
\end{split}
\end{equation}
where the second line is expressed in terms of Choi operators of the process and tester (see Appendix~\ref{app:choi-combs-testers}).
In this sense, testers implement generalized measurements on processes. 

Testers can be realized by a sequence of CPTP maps $\{\mathcal{T}^{(k)}\}_{k=0}^{N}$ acting on the system and memory, but connected only via memory and are followed by a final quantum instrument $\{\mathcal{M}^{(N+1)}_{j}\}_j$. A tester element can thus be written as,
\begin{align}\label{eqn:tester_as_sequence}
\mathbf{T}^{N}_j
:=
\tr_{M_{N+2}}\!\left(
\mathcal{M}_{j}^{(N+1)}\!\circ_{M_{N+1}}\! \mathcal{T}_{}^{(N)}\!\circ_{M_{N}}\!\dots\circ_{M_{0}}\! \mathcal{T}_{}^{(0)}[\rho^{\mathbf{T}}_{M,0}]
\right),
\end{align}
where $\circ_{M_{k}}$ denotes composition over the memory space at the $k^{\text{th}}$ time step (see Fig.~\ref{fig:tester_and_comb}). The state $\rho^{\mathbf{T}}_{M,0}$ is the initial memory state specified by the tester. This mirrors the internal structure of the process to which it is applied.

Once the maps, $\{\mathcal{T}^{(k)}\}$ have acted in the slots of the process (at intermediate times between maps $\{\mathcal{P}^{(k)}\}$), the quantum instrument $\{\mathcal{M}^{(N+1)}_{j}\}_j$ acts on a final system-memory output state $\rho_{SM,N}$ as $\mathcal{J}[\rho_{SM,N}]:=\sum_j\mathcal{M}_j[\rho_{SM,N}]\otimes\dyad{j}$ such that, upon observing the outcome $j$, $\mathcal{J}$ updates a classical register to $\dyad{j}$ and the state $\rho_{SM,N}$ via the completely positive (CP) map $\mathcal{M}_j$~\cite{vieira_temporal_2022, Zonnios2025QuantumGeneration}. The instrument also induces a positive operator-valued measure (POVM) $\{E_j\}$ with elements $E_j=\sum_\alpha K_\alpha^{(j)\dag} K_\alpha^{(j)}$ specified for any choice of Kraus operators $K_\alpha$ of the corresponding map $\mathcal{M}_j$, $\mathcal{M}_j[\rho]=\sum_\alpha K_{\alpha}^{(j)}\rho K_\alpha^{(j)\dag}$. The probability of obtaining outcome $j$ when the system is in the state $\rho_{SM,N}$ follows the Born rule, $\Pr(j|{\rho_{SM,N}})=\tr(E_j\rho_{SM,N})$. 

Although one may consider performing measurements at multiple time steps, this construction allows all intermediate measurements to be deferred. In other words, any sequence of instruments applied during the process can be equivalently represented by a single measurement at the end, whose outcomes $j$ label the entire sequence of intermediate outcomes (see App.~\ref{app:testers}). Consequently, the number of possible outcomes $x_{0:N}:=(x_0,\dots,x_{N})$ grows as the product of the number of outcomes at each time step,  $\#\text{outcomes}=\prod_{k=0}^N|\mathcal{X}_k|$ (for output alphabets $\mathcal{X}_k$ at times $k$), and is therefore generally exponential in $N$ when the dimension of the system $d_{S_j}$ is fixed for all times $j$.

Testers thus provide the most general measurement framework for multi-time processes. In particular, they enable the formulation of discrimination tasks between quantum processes, which we quantify using the strategy-norm distance.

\subsection{The strategy-norm distance between quantum processes}

For quantum states, the standard operational measure of distinguishability is the trace distance, $\tfrac12\|\rho-\sigma\|_1$. It characterizes the optimal bias in binary state discrimination -- that is, the maximum advantage over random guessing in deciding whether the unknown state is $\rho$ or~$\sigma$. For quantum channels, the analogous measure is the diamond-norm distance, $\tfrac12\|\mathcal{P}-\mathcal{Q}\|_\diamond$, where the optimization allows for arbitrary input states, including ancillary systems, and final positive operator valued measure (POVMs)~\cite{gutoski2012measure,watrous2018theory}. The natural extension to multi-time quantum processes is the strategy-norm distance~\cite{gutoski2012measure,watrous2018theory}.

For two $N$-step processes $\mathbf{P}^{N}$ and $\mathbf{Q}^{N}$, this distance is defined as
\begin{equation}
\label{eq:strategy-norm}
\begin{split}
d_{\mathrm{str}}^{(N)}(\mathbf{P}^N,\mathbf{Q}^N)
&:=
\frac12|\mathbf{P}^N-\mathbf{Q}^N|_{\mathrm{str}}
\\&=
\sup_{\mathbf{T}^{N}_{0|1}}
\Big|\Pr(0|\mathbf{P}^{N})-\Pr(0|\mathbf{Q}^{N})\Big|
\\
&=\;
\sup_{\Upsilon^{\mathbf{T}^{N}_{0|1}}}\,
\Big|\mathrm{tr}\!\left((\Upsilon^{\mathbf{P}^{N}}-\Upsilon^{\mathbf{Q}^{N}})\,
(\Upsilon^{\mathbf{T}^{N}_{0|1}})^T\right)\Big|,
\end{split}
\end{equation}
where the optimization is over all adaptive probing strategies, equivalently over all testers, with a final binary outcome. The tester $\{\mathbf{T}^{N\star}_{0|1}\}$ achieving the strategy norm distance is referred to as the optimal tester~\cite{gutoski2012measure}. Although a tester may have many outcomes, binary discrimination only requires a binary decision. For any fixed tester, the optimal strategy assigns each outcome to the hypothesis with larger likelihood, inducing a coarse-graining of outcomes into two subsets. Therefore, the optimization can be restricted to two-outcome testers without loss of generality (see App.~\ref{app:binary_testers} for details). For $N=0$, this quantity reduces to the trace distance, $\tfrac12\|\rho-\sigma\|_1$, while for $N=1$ it reduces to the diamond-norm distance, $\tfrac12\|\mathcal{P}-\mathcal{Q}\|_\diamond$ (see App.~\ref{app:strategy-norm-reduces-trace-diamond} and Refs.~\cite{gutoski2012measure,watrous2018theory} for details). Moreover, it is contractive under physically valid superprocesses \cite{ZambonDistinguishability2024}.

The optimization in Eq.~\eqref{eq:strategy-norm} admits a semidefinite programming (SDP) formulation \cite{gutoski2012measure,Watrous2009SDP}, but the dimension of the underlying operator space grows exponentially with the number of time steps, as do the number of causality constraints imposed on the testers for increasing $N$. In particular, using state-of-the-art interior-point methods runtime typically grows as $\mathcal{O}(d_S^{12N})$ per iteration~\cite{Watrous2009SDP,Jiang2020AFI, mironowicz2024sdp} (see App.~\ref{app:strategy-norm-sdp} for details) where $d_S$ is the system dimension. This motivates restricting the class of admissible testers to those compatible with realistic experimental constraints, leading to the $\mathsf{MAD}$ framework introduced next. Specifically, we can overcome this unfavorable scaling by introducing a memory-constrained variant of the strategy-norm distance, based on restricted testers that approximate the optimal bias while alleviating the exponential complexity of the underlying SDP. This leads to substantially faster evaluation and facilitates simpler diagnostics for processes generated by repeated interactions.

\section{Distinguishability of quantum processes under constrained memory}
\label{sec:MAD-definition}

\subsection{Machine for autonomous distinction ($\mathsf{MAD}s$) and $\mathsf{MAD}$ testers }

\begin{figure}
    \centering
    \includegraphics[width=1\linewidth]{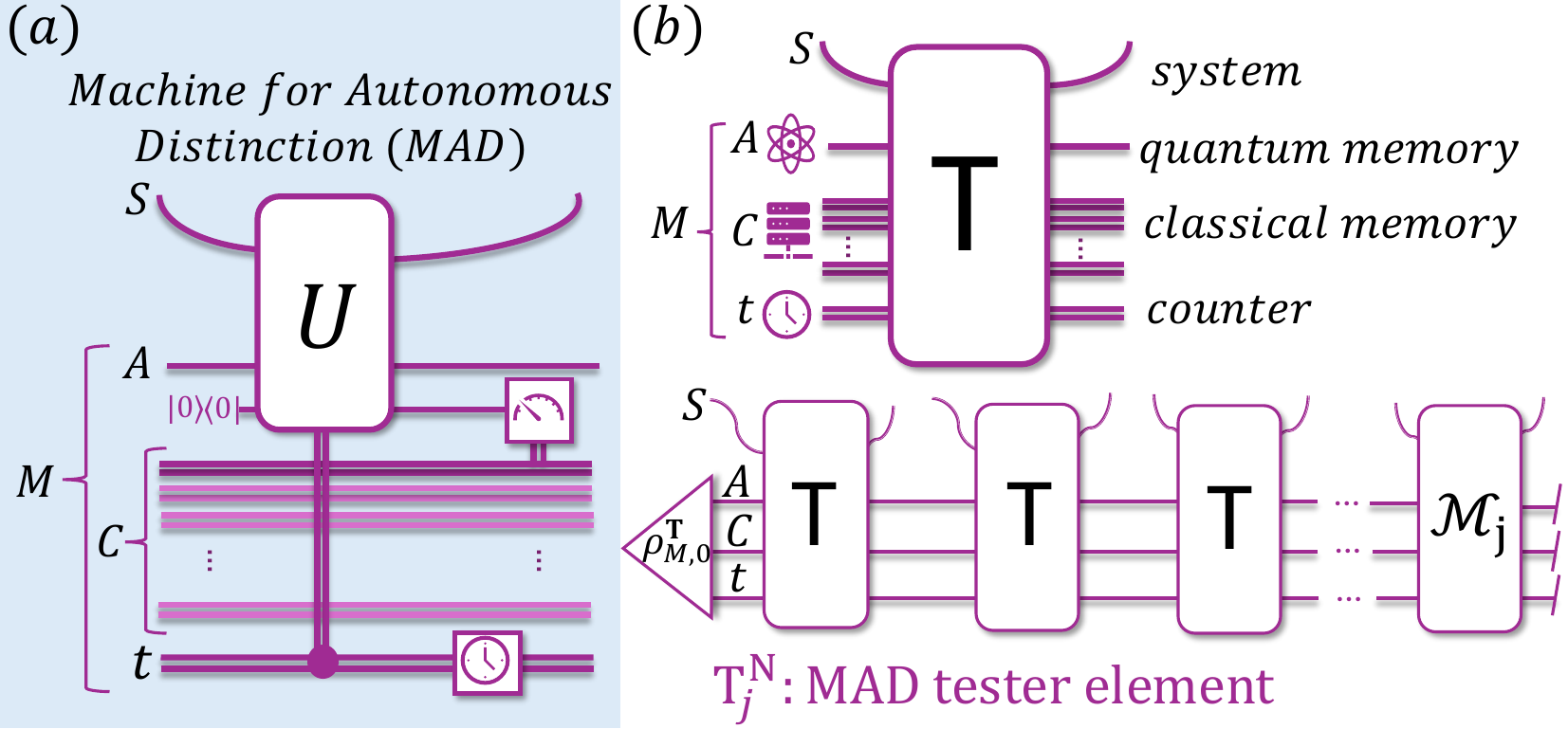}
    \caption{
(a) A \textit{machine for autonomous distinction} ($\mathsf{MAD}$) is a probing device modelled as a quantum instrument acting on the system and a coherent ancillary memory. Each use of the instrument can be realized by a unitary interaction with an auxiliary probe followed by a projective measurement on the probe, its nonselective action being CPTP. In the autonomous implementation shown here, the unitary acts on the system and quantum memory, while a classical counter specifies which output register stores the outcome and, if desired, allows explicit time dependence to be encoded internally.
(b) A $\mathsf{MAD}$ tester is obtained by repeatedly applying the same $\mathsf{MAD}$ to the process slots as a special case of $\mathbf{T}^N_j$ from Fig.~\ref{fig:tester_and_comb}. The tester carries a coherent quantum memory of dimension $d_A$, together with a classical memory large enough to store the outcome record for $N+1$ iterations and a counter that labels the time step. After the repeated interaction, a final distinguishing measurement via the quantum instrument $\{\mathcal{M}_j\}_j$ is performed on the accessible output. This defines a $\mathsf{MAD}$ tester for distinguishing between two processes $\mathbf{P}^N$ and $\mathbf{Q}^N$.
}
    \label{fig:MAD}
\end{figure}

We consider finite-dimensional internal memory consisting of a quantum ancilla of dimension $d_A$ together with classical registers used to store measurement outcomes. This construction reflects the fact that maintaining coherent quantum memory is typically costly, whereas classical information can be stored and propagated efficiently over many time steps~\cite{preskill2018nisq,taranto2021}. Recent work identifies memory as a fundamental resource in multi-time processes and characterizes hierarchies in terms of classical and quantum memory~\cite{taranto2024hierarchy,giarmatzi2023}. Our formulation exploits the fact that while optimal discrimination of quantum processes may in general require an exponentially large coherent memory~\cite{GutoskiWatrous2007,chiribella2008}, storing classical outcomes together with only a bounded amount of coherent memory can already yield good approximations to the full strategy-norm distance.

\begin{definition}[Machines for Autonomous Distinction]\needspace{3\baselineskip}
A \textit{machine for autonomous distinction} (a $\mathsf{MAD}$) is a quantum instrument that acts jointly on a system and coherent quantum memory of dimension $d_A$,
\begin{equation}
\label{eq:MAD-instrument}
\mathsf{T} \equiv \left(\{\mathcal{T}^{(x_k)}\}_{x_k,k}, N\right),
\end{equation}
where, for each $k=0,\ldots,N$, the maps ${\mathcal{T}^{(x_k)}}{x_k}$ form a valid quantum instrument. The $\mathsf{MAD}$ has access to a classical register with $N+1$ slots, each with dimension $|\mathcal{X}_k|$, to keep track of the outcomes at all integer times $k=0,\ldots,N$. The state of the counter determines which quantum instrument is applied and which slot of the classical register is updated. 
\end{definition}

Intuitively, a $\mathsf{MAD}$ can be seen as a fixed operation, device or quantum instrument that is used repeatedly at each time step, with all adaptivity encoded internally via memory registers. The overall evolution $\mathsf{T}$ remains fixed and identical at every step, while the counter determines the step index $k$ and hence the instrument $\{\mathcal{T}^{(x_k)}\}_{x_k}$ applied at that step. The classical register states are initialized as $\dyad{0}^{\otimes (N+1)}C$ and a time-keeping classical state is also initialized as $\dyad{0}t$. The quantum ancilla is initialized as $\rho_{A,0}$ and evolves coherently under the tester dynamics along with the initial state of the system, which is possibly entangled with the environment. The total initial state that we consider is then $\rho_{ESA,0}$ together with the counter and classical register. For the first application, the $\mathsf{MAD}$ acts on this state as
\begin{align}
\label{eq:mad-first-step}
\begin{split}
&\mathsf{T}\!\otimes\!\mathds{1}_E
\big[\rho_{ESA,0}\otimes\dyad{0}_t\!\otimes\!\dyad{0}^{\otimes (N+1)}_C\big]\
\\&\quad :=
\sum_{x_0\in\mathcal{X}_0}
\mathcal{T}^{(x_0)}\!\otimes\!\mathds{1}_E[\rho_{ESA,0}]
\!\otimes\!\dyad{1}_t\!\otimes\!\dyad{x_0}_C\!\otimes\!\dyad{0}_C^{\otimes N}
\
\\&\quad =
\sum_{x_0\in\mathcal{X}_0}
\hat{\rho}_{ESA,x_0}
\!\otimes\!\dyad{1}_t\!\otimes\!\dyad{x_0}_C\!\otimes\!\dyad{0}_C^{\otimes N}.
\end{split}
\end{align}

Here $\mathds{1}_E$ denotes the identity over the inaccessible environment and $\hat{\rho}_{ESA,x_0}$ denotes an $x_0$ output-dependent sub-normalised state on the system, environment and ancillary spaces. For each fixed $k$, the maps $\{\mathcal{T}^{(x_k)}\}_{x_k}$ form a valid quantum instrument, i.e. $\sum_{x_k}\mathcal{T}^{(x_k)}$ is completely positive and trace-preserving (CPTP) and thus $\tr\!\left(\sum_{x_0}\hat{\rho}_{ESA,x_0}\right)=1$. Consequently, the repeated composition of the corresponding nonselective maps is also CPTP. 

Operationally, a $\mathsf{MAD}$ is simply a quantum instrument whose classical outcomes are recorded in memory at each time step, while the system and ancilla evolve coherently. Explicit time dependence can be incorporated internally via the step index $k$, so that the overall evolution $\mathsf{T}$ remains fixed and identical at every step. Equivalently, a $\mathsf{MAD}$ can be realized by a fixed unitary interaction with an auxiliary probe, followed by a measurement, together with classical registers that store the outcomes as shown in Fig.~\ref{fig:MAD}. Besides the non-adaptive case in which the same instrument is applied at each step with fixed coherent memory, $\mathsf{MAD}$ testers can also represent adaptive testers with fixed coherent ancillary memory. Explicit time dependence can be encoded internally by the counter register, so that the operation applied at step $k$ may correspond to a different instrument, or equivalently to a different unitary dilation on the system and ancillary space, while the coherent memory dimension remains fixed. This construction, established formally in Sec.~\ref{sec:mad-finite-time-completeness}, shows that any finite-time tester can be implemented as a $\mathsf{MAD}$ tester given sufficiently large coherent memory.

Applying the same $\mathsf{MAD}$ repeatedly via Eq.~\eqref{eq:mad-first-step} to the outputs of a hypothesis process $\mathbf{H}^{N}\in\{\mathbf{P}^{N},\mathbf{Q}^{N}\}$ (as in Fig.~\ref{fig:tester_and_comb}) and tracing out the environment and counter yields the classical-quantum state
\begin{align}
\label{eqn:tester_process_state_evolution}
\begin{split}
\rho^{\mathsf{T},\mathbf{H}}_{SAC,N}
:=&
\tr_{E_{N+1},t}\big({H}^{(N)}\!\circ\!\mathsf{T}\!\circ\!\dots
\!\circ\!{H}^{(1)}\circ\mathsf{T}[\rho_{SE,0}^{\mathbf{H}} \otimes\rho_{M,0}^{\mathsf{T}}]\big)
\\=&
\sum_{x_{0:N}\in\mathcal{X}^{N+1}}\tr_{E_{N+1}}\big(\hat{\rho}_{ESA,x_{0:N}}^{\mathsf{T},\mathbf{H}}\big)\!\otimes\! \dyad{x_{0:N}}_C,
\end{split}
\end{align}
where $M$ is the total memory (coherent ancillary $A$, classical register $C$ and counter $t$) which the $\mathsf{MAD}$ $\mathsf{T}$ acts on; the initial memory state is given by $\rho_{M,0}^{\mathsf{T}}:=\rho_{A,0}^{\mathsf{T}}\otimes\dyad{0}_t\otimes\dyad{0}^{\otimes {N+1}}_C$; the maps $\{H^{(k)}\}_k$ act on the system and environment at time $t_k$ corresponding to the process $\mathbf{H}^N$; and $\rho_{SE,0}^{\mathbf{H}}$ is the hypothesis-specified initial state on the system and environment. At each step, $\mathsf{T}$ acts on the system and memory, while the hypothesis map acts on the system and environment (the identities on the environment $\mathds{1}_E$ and $\mathds{1}_M$ memory respectively are suppressed in Eq.~\eqref{eqn:tester_process_state_evolution}). The environment is updated by the process, the coherent memory is updated by the $\mathsf{MAD}$, and the system is updated by both. A classical outcome $x_k$ is appended to a classical register that is stored for processing at some designated final time. After $N$ steps, the induced output under hypothesis $\mathbf{H}^{N}$ is given by tracing out the environment and the counter to give the classical-quantum state
\begin{equation}
\label{eq:cq-output}
\rho_{SAC,N}^{\mathsf{T},\mathbf{H}}
=
\sum_{x_{0:N}\in\mathcal{X}^{N+1}}
|x_{0:N}\rangle\!\langle x_{0:N}|_{C}\otimes \hat{\sigma}^{\mathsf{T},\mathbf{H}}_{x_{0:N}},
\end{equation}
where $|x_{0:N}\rangle$ is the length-$N$ record and $\hat{\sigma}^{\mathsf{T},\mathbf{H}}_{x_{0:N}}=\mathbf{P}^{N}[\mathbf{T}^{N}_{x_{0:N}}]$ is the corresponding subnormalized final state on the system and quantum ancillary space corresponding to the outcome sequence $x_{0:N}\equiv(x_0,x_1,\dots,x_N)$ after tracing out the inaccessible environment. The set of output states for a given tester satisfy, 
\begin{align}
\begin{split}
    \sum_{x_{0:N}}\tr(\hat{\sigma}^{\mathsf{T},\mathbf{H}}_{x_{0:N}} )&=\sum_{x_{0:N}} \tr\left(\mathbf{H}^{N}[\mathbf{T}^{N}_{x_{0:N}}]\right) \\&= \sum_{x_{0:N}} \Pr(x_{0:N}|\mathbf{H}^N,\mathbf{T}^N) \\&= 1.
\end{split}
\end{align}

For a fixed $\mathsf{MAD}$ tester, the two hypotheses $\mathbf{P}^{N}$ and $\mathbf{Q}^{N}$ induce two final output states $\rho_{N}^{\mathsf{T},\mathbf{P}}$ and $\rho_{N}^{\mathsf{T},\mathbf{Q}}$ where we have suppressed the label $SAC$. The discrimination
problem for the processes therefore reduces to the standard binary discrimination problem for these two states. With priors $p_0$ and $q_0$, the optimal measurement is the Helstrom
measurement~\cite{helstrom1969}, given by the positive and negative eigenspaces of $p_0\rho_{N}^{\mathsf{T},\mathbf{P}}-q_0\rho_{N}^{\mathsf{T},\mathbf{Q}}$. In the following we take equal priors, $p_0=q_0=1/2$. The optimal success probability for distinguishing the two processes using the repeated $\mathsf{MAD}$ $\mathsf{T}$ is then
\begin{align}
\label{eq:helstrom-success}  
p_{\mathrm{succ}}^*(\mathbf{P}^{N},\mathbf{Q}^{N},\mathsf{T})
&= \frac12\Big(1+\frac12\big\|\rho^{\mathsf{T},\mathbf{P}}_{N}-\rho^{\mathsf{T},\mathbf{Q}}_{N}\big\|_1\Big).
\end{align}

\begin{definition}[$\mathsf{MAD}$ testers]
A length-$N$ \textit{$\mathsf{MAD}$ tester}~$(\mathsf{T},\rho_{M,0}^{\mathsf{T}})$ is a quantum tester obtained by specifying an initial quantum ancillary state~$\rho_{A,0}^{\mathsf{T}}$, and consequently the initial memory state $\rho_{M,0}^{\mathsf{T}}:=\rho_{A,0}^{\mathsf{T}}\otimes\dyad{0}_t\otimes\dyad{0}^{\otimes N}_C$, applying the same $\mathsf{MAD}$ $\mathsf{T}$ for $N$ time steps, with memory connecting successive uses, followed by a final measurement. For the binary discrimination task between $\mathbf{P}^{N}$ and $\mathbf{Q}^{N}$, the final measurement is the Helstrom measurement between the respective states $\rho^{\mathsf{T},\mathbf{P}}_{N}$ and $\rho^{\mathsf{T},\mathbf{Q}}_{N}$ generated by Eqs.~\eqref{eqn:tester_process_state_evolution} and~\eqref{eq:cq-output}.
\end{definition}

We denote the class of all $\mathsf{MAD}$ testers with coherent memory dimension $d_A$ by
\begin{align}
\label{eq:mad-class}
\mathsf{MAD}(d_A)
:=
\{(\mathsf{T},\rho_{A,0}^{\mathsf{T}}) \text{ with } \dim(A)=d_A\}.
\end{align}

In practice, $\mathsf{MAD}$ testers are autonomous process discrimination strategies for quantum processes that repeatedly apply the same $\mathsf{MAD}$ at each time step without external time-dependent control, so all adaptivity is mediated by internal registers. The classical register stores outcomes, while the coherent memory propagates quantum information and correlations between steps, inducing a natural resource hierarchy indexed by the coherent memory dimension $d_A$. Operationally, this restriction isolates memory as the resource governing temporal correlations while allowing unrestricted classical post-processing~\cite{holevo2011,bae2015}. 

The MAD construction therefore identifies a restricted, but operationally meaningful, class of probing strategies in which the probing strategy may retain an unrestricted classical record, but only a bounded amount of coherent quantum memory. This restriction is useful for two reasons. First, it separates the information stored as classical outcomes from the information carried coherently between interventions. The latter is precisely the resource that allows a tester to access quantum temporal correlations, and its dimension controls which recovery, simulation, and discrimination tasks can be performed on multi-time processes~\cite{milz2020,taranto2021,taranto2024hierarchy,ohst2026}. Second, it provides a practical alternative to optimizing over the full class of testers. Although the strategy-norm distance gives the optimal distinguishing bias, evaluating it requires an optimization over general adaptive testers whose memory dimension and constraint structure grow as $\sim d_S^{2N}$ and $\sim d_S^{4N}$, respectively, where $d_S$ is the system dimension and $N$ is the number of interventions. The solution admits a semidefinite programming formulation which, even with state-of-the-art interior-point methods, leads to a runtime scaling of order $\mathcal{O}(d_S^{12N})$~\cite{Jiang2020AFI}. This quickly becomes impractical beyond small $N$. The natural next step is therefore to define a distinguishability measure obtained by optimizing only over MAD testers with a fixed coherent memory dimension $d_A$, which we do in the following section.

\subsection{$\mathsf{MAD}$ distinguishability}
\label{sec:mad_distinguishability}

We now introduce a memory-constrained version of the strategy-norm distance between quantum processes. First, we show that the measure defines a bounded operational pseudometric on processes. Then, we show that it forms a hierarchy in the coherent memory dimension and is always upper bounded by the full strategy-norm distinguishability.

\begin{definition}[$\mathsf{MAD}$ distinguishability]\needspace{3\baselineskip}
Fix any two processes $\mathbf{P}^{N}$ and $\mathbf{Q}^{N}$. The \textit{$\mathsf{MAD}$ distinguishability} between the two processes is the optimal bias achievable by any $\mathsf{MAD}$ tester with coherent memory dimension $d_A$,
\begin{equation}
\label{eq:drec-def}
d_{\mathsf{MAD}}^{(N)}(\mathbf{P}^{N},\mathbf{Q}^{N};d_A)
:=
\sup_{(\mathsf{T},\rho_{M,0}^{\mathsf{T}})\in \mathsf{MAD}(d_A)}
\frac12 \big\| \rho^{\mathsf{T},\mathbf{P}}_{N} - \rho^{\mathsf{T},\mathbf{Q}}_{N} \big\|_1,
\end{equation}
where the respective states $\rho^{\mathsf{T},\mathbf{P}}_{N}$ and $\rho^{\mathsf{T},\mathbf{Q}}_{N}$ are generated by Eqs.~\eqref{eqn:tester_process_state_evolution} and~\eqref{eq:cq-output}.
\end{definition}

\begin{proposition}[Pseudometric properties of $\mathsf{MAD}$ distinguishability]
\label{prop:basic-properties-dmad}
For fixed $N$ and $d_A$, $d_{\mathsf{MAD}}^{(N)}(\cdot,\cdot;d_A)$ defines a bounded operational pseudometric on processes. In particular, for any processes $\mathbf{P}^{N}$, $\mathbf{Q}^{N}$, and $\mathbf{R}^{N}$,
\begin{enumerate}
    \item \textit{Boundedness:}
    \[
    0\le d_{\mathsf{MAD}}^{(N)}(\mathbf{P}^{N},\mathbf{Q}^{N};d_A)\le 1.
    \]
    \item \textit{Symmetry:}
    \[
    d_{\mathsf{MAD}}^{(N)}(\mathbf{P}^{N},\mathbf{Q}^{N};d_A)
    =
    d_{\mathsf{MAD}}^{(N)}(\mathbf{Q}^{N},\mathbf{P}^{N};d_A).
    \]
    \item \textit{Triangle inequality:}
    \[
    \begin{split}
    \label{eq:dmad-triangle}
    d_{\mathsf{MAD}}^{(N)}(\mathbf{P}^{N},\mathbf{R}^{N};d_A)
    \le
    &d_{\mathsf{MAD}}^{(N)}(\mathbf{P}^{N},\mathbf{Q}^{N};d_A)\\
    &+
    d_{\mathsf{MAD}}^{(N)}(\mathbf{Q}^{N},\mathbf{R}^{N};d_A).
    \end{split}
    \]
\end{enumerate}
\end{proposition}

\begin{proof}
Boundedness follows since $d_{\mathsf{MAD}}^{(N)}$ is one half of a trace distance between two final classical--quantum states, hence lies in $[0,1]$. Symmetry is immediate from $\|\rho-\sigma\|_1=\|\sigma-\rho\|_1$. For the triangle inequality, fix any $\mathsf{MAD}$ tester $\mathsf{T}$ and let $\rho^{\mathsf{T},\mathbf{H}}_{N}$ denote the corresponding final classical--quantum state under hypothesis $\mathbf{H}^{N}$. Then
\begin{align}
\frac12\|\rho^{\mathsf{T},\mathbf{P}}_{N}-\rho^{\mathsf{T},\mathbf{R}}_{N}\|_1
\le
\frac12\|\rho^{\mathsf{T},\mathbf{P}}_{N}-\rho^{\mathsf{T},\mathbf{Q}}_{N}\|_1
+
\frac12\|\rho^{\mathsf{T},\mathbf{Q}}_{N}-\rho^{\mathsf{T},\mathbf{R}}_{N}\|_1
\end{align}
by the triangle inequality for the trace distance. Taking the supremum over $(\mathsf{T},\rho_{M,0}^{\mathsf{T}})\in\mathsf{MAD}(d_A)$ and using $\sup_T(A_T+B_T)\le \sup_T A_T+\sup_T B_T$ yields property~\ref{eq:dmad-triangle}.
\end{proof}

Consequently, for fixed $N$ and $d_A$, $d_{\mathsf{MAD}}^{(N)}(\cdot,\cdot;d_A)$ satisfies nonnegativity, symmetry, and the triangle inequality, but may assign zero distance to distinct processes that are indistinguishable within the $\mathsf{MAD}(d_A)$ class.

\begin{proposition}[Hierarchy properties]
\label{prop:mad-hierarchy-properties}
The $\mathsf{MAD}$ distinguishability satisfies:
\begin{enumerate}
    \item \textit{Monotonicity in coherent memory:} If $d_A'\ge d_A$, then
    \[
    d_{\mathsf{MAD}}^{(N)}(\mathbf{P}^{N},\mathbf{Q}^{N};d_A')
    \ge
    d_{\mathsf{MAD}}^{(N)}(\mathbf{P}^{N},\mathbf{Q}^{N};d_A).
    \]
    \item \textit{Monotonicity in time:} If $\mathbf{P}^{N}$ and $\mathbf{Q}^{N}$ are the $N$-step restrictions of consistent $(N+1)$-step processes $\mathbf{P}^{N+1}$ and $\mathbf{Q}^{N+1}$, then
    \[
    d_{\mathsf{MAD}}^{(N)}(\mathbf{P}^{N},\mathbf{Q}^{N};d_A)
    \le
    d_{\mathsf{MAD}}^{(N+1)}(\mathbf{P}^{N+1},\mathbf{Q}^{N+1};d_A).
    \]
    \item \textit{Upper bound by the strategy-norm distance:} For all $d_A$,
    \[
    \label{eq:dmad-leq-dstr}
    d_{\mathsf{MAD}}^{(N)}(\mathbf{P}^{N},\mathbf{Q}^{N};d_A)
    \le
    d_{\mathrm{str}}^{(N)}(\mathbf{P}^{N},\mathbf{Q}^{N}),
    \]
    where $d_{\mathrm{str}}^{(N)}(\mathbf{P}^{N},\mathbf{Q}^{N})$ denotes the optimal distinguishability over the full class of admissible testers, equivalently the strategy-norm-defined distinguishability.
\end{enumerate}
\end{proposition}

\begin{proof}
For monotonicity in $d_A$, embed any $\mathsf{MAD}$ tester with coherent memory dimension $d_A$ into one with dimension $d_A'\ge d_A$; this enlarges the optimization domain and cannot reduce the supremum. The upper bound by the strategy-norm distance holds because $\mathsf{MAD}(d_A)$ is a subset of the full class of admissible testers. For monotonicity in time, any $N$-step $\mathsf{MAD}$ tester can be regarded as an $(N+1)$-step tester that ignores the additional step, or equivalently discards the final-step registers. Since trace distance is monotonic under CPTP maps~\cite{Nielsen_Chuang_2010,watrous2018theory}, the optimal distinguishing advantage cannot decrease when an additional process time is made available.
\end{proof}

\subsection{Finite-time completeness of the $\mathsf{MAD}$ hierarchy}
\label{sec:mad-finite-time-completeness}

A natural question is whether the $\mathsf{MAD}$ restriction reduces expressive power relative to general multi-time testers. This is addressed by work on constrained discrimination strategies~\cite{nakahira2021restricted,ohst2026}, which shows that performance is governed by quantum memory rather than control complexity. The next result establishes that, for any finite time, the $\mathsf{MAD}$ restriction is fully general in that any tester can be realized as repeated use of a single recurrent $\mathsf{MAD}$ with sufficiently large coherent memory.

\begin{theorem}[Compilation of any tester into a $\mathsf{MAD}$ tester]
\label{thm:mad-compilation}
Fix a time $N$ and let $\{\mathbf{T}^{N}_j\}_j$ be any admissible $N$-step tester. Then there exists a $\mathsf{MAD}$ tester $(\mathsf{T},\rho_{A,0}^{\mathsf{T}})$ such that, for every $N$-step process $\mathbf{P}^{N}$, the induced outcome statistics coincide,
\begin{equation}
\label{eq:mad-compilation-stats}
\Pr(j \mid \mathbf{P}^{N}, \{\mathbf{T}^{N}_j\}_j)
=
\Pr(j \mid \mathbf{P}^{N}, (\mathsf{T},\rho_{A,0}^{\mathsf{T}}))
\qquad \forall\, j .
\end{equation}
\end{theorem}

\begin{proof}
By the standard sequential realisation of testers~\cite{GutoskiWatrous2007,chiribella2009theoretical}, any $N$-step tester can be implemented as a sequence of CPTP maps linked by a finite-dimensional quantum memory, followed by a final measurement. Introducing a classical counter register that selects the appropriate step map allows this sequence to be implemented by repeated application of a single identical quantum instrument. A detailed construction is given in Appendix~\ref{app:mad-simulates-tester}.
\end{proof}

Theorem~\ref{thm:mad-compilation} shows that, for any fixed finite time $N$, restricting to recurrent process discrimination strategies does not reduce the expressive power of multi-time testers. Any admissible tester can be implemented by repeated use of an identically repeated quantum instrument whose outcomes are stored in a classical memory register, with all time dependence encoded in an internal classical counter and correlations mediated by a coherent quantum memory.

\begin{corollary}[Finite-time completeness and memory bound]
\label{cor:mad-saturates-strategy}
\label{cor:mad-dimension}
Fix two $N$-step processes $\mathbf{P}^{N}$ and $\mathbf{Q}^{N}$. Then there exists a finite $d_A^\star$ such that for all $d_A\ge d_A^\star$,
\begin{equation}
\label{eq:mad-saturates}
d_{\mathsf{MAD}}^{(N)}(\mathbf{P}^{N},\mathbf{Q}^{N};d_A)
=
d_{\mathrm{str}}^{(N)}(\mathbf{P}^{N},\mathbf{Q}^{N}).
\end{equation}
Moreover, $d_A^\star$ may be chosen no larger than the largest intermediate quantum memory dimension appearing in a sequential realisation of an optimal tester,
\begin{equation}
\label{eq:mad-compilation-dim}
d_A^\star \le d_{\mathrm{suf}},
\qquad
d_{\mathrm{suf}} := \max_{0\le k\le N}\dim \mathcal H_{M_k}.
\end{equation}
\end{corollary}
\begin{proof}
Let $\{\mathbf{T}^{N\star}_j\}_j$ be a tester achieving $d_{\mathrm{str}}^{(N)}(\mathbf{P}^{N},\mathbf{Q}^{N})$ as per Eq.~\eqref{eq:strategy-norm}. By Theorem~\ref{thm:mad-compilation}, for $d_A\ge d_A^\star$ there exists a $\mathsf{MAD}$ tester whose outcome statistics match those of $\{\mathbf{T}^{N*}_j\}_j$ on both hypotheses. Hence there is a final measurement on the $\mathsf{MAD}$ output that reproduces the optimal tester's binary decision statistics. The Helstrom measurement on the $\mathsf{MAD}$ output cannot do worse, so
\[
d_{\mathsf{MAD}}^{(N)}(\mathbf{P}^{N},\mathbf{Q}^{N};d_A)
\ge
d_{\mathrm{str}}^{(N)}(\mathbf{P}^{N},\mathbf{Q}^{N}).
\]
The reverse inequality holds because $\mathsf{MAD}(d_A)$ is a subset of the full tester class, as in Eq.~\eqref{eq:dmad-leq-dstr}. Combining both inequalities yields Eq.~\eqref{eq:mad-saturates}. The memory bound follows from the sequential realization used in Theorem~\ref{thm:mad-compilation} -- the compiled $\mathsf{MAD}$ only needs a coherent ancilla large enough to contain the largest intermediate memory space $\mathcal H_{M_k}$ of the tester realization. 
\end{proof}

For uniform input and output system dimension $d_S$, a coarse bound is obtained by bounding this memory by the Hilbert-space dimension of all $N$ input-output slots, giving $d_{\mathrm{suf}}\le d_S^{2N}$. Corollary~\ref{cor:mad-saturates-strategy} establishes that the $\mathsf{MAD}$ hierarchy is complete at finite times. Increasing the coherent memory dimension $d_A$ suffices to recover the full strategy-norm distinguishability. The bound $d_{\mathrm{suf}}\le d_S^{2N}$ is generally loose as the necessary quantum memory is controlled by the largest intermediate memory space appearing in a sequential realisation of the tester and may be much smaller for some choices of processes.

While in principle $d_{\mathrm{suf}}$, and hence $d_A^\star$, may grow exponentially with $N$, as we will show in Sec.~\ref{sec:numerics}, much smaller coherent memory often suffices in practice to approximate, and in some cases even attain, the strategy-norm distance. In particular, the required memory can be significantly smaller than that implicitly used in standard semidefinite programming approaches.

This reduction in coherent memory has significant computational implications. The dimension of the optimization space for $\mathsf{MAD}$ testers scales directly with $d_A$, so restricting to small memory dramatically reduces the effective problem size. In contrast, direct semidefinite programming approaches scale with the full process dimension and quickly become intractable. Thus, memory-limited $\mathsf{MAD}$ optimization provides a practical route to approximating the strategy-norm distance beyond the reach of standard SDP methods.

To gain further analytical insight, we now specialize to a physically relevant class of processes generated by repeated and identical system--environment interactions.
\section{$\mathsf{MAD}$ testers for repeated interaction processes}\label{sec:repeated_interaction_picture}

\begin{figure}
    \centering
    \includegraphics[width=0.9\linewidth]{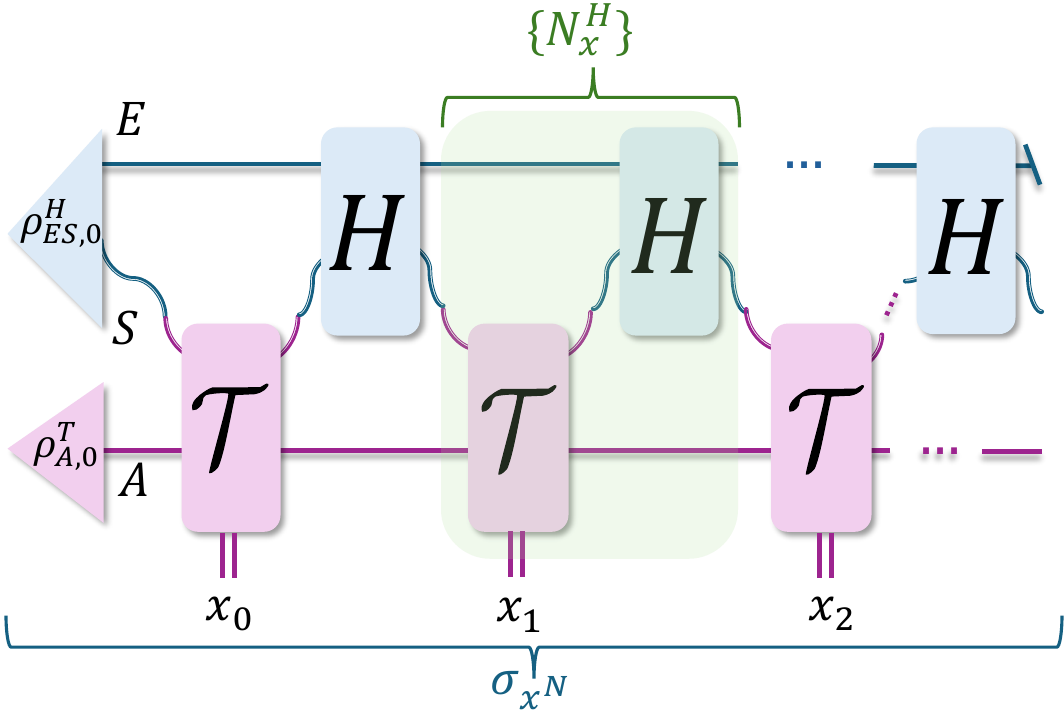}
    \caption{
A recurrent process hypothesis probed by a $\mathsf{MAD}$ tester. The recurrent process $\mathbf{H}_H$ (blue) is obtained by repeatedly applying $H$ which interacts with a persistent environment $E$, while a fixed quantum instrument $\{\mathcal{T}^{(x)}\}_{x}$ (purple) acts jointly on the system and quantum memory $A$ at each time step. The instrument produces a classical outcome $x_k$ that is stored, generating a trajectory $x_{0:N}$, while the system and memory evolve coherently. The induced single-step maps $\{\mathcal{N}_x^{\mathsf{T},H}\}_x$ (highlighted in green) also correspond to a quantum instrument and describe the joint action of the process and tester at each step, and generate the branchwise states $\hat{\sigma}_{x_{0:N}}^{\mathsf{T},H}$.
}
    \label{fig:reccurent_hypothess_plus_tester}
\end{figure}

We now specialize to recurrent processes generated by repeating a fixed system-environment interaction. A hypothesis $H\in\{P,Q\}$ is specified by a finite-dimensional environment space, an initial joint state $\rho^{H}_{SE}$ on system and environment, and a fixed CPTP map applied at each time step. A recurrent process is then given by, 
\begin{align}
\mathbf{H}_H^{N}
:=
\tr_{E_{N+1}}\!\left(
{H}\circ_{E_N}\dots\circ_{E_3}{H}\circ_{E_2} {H}\circ_{E_1}[\rho_{SE,0}^{\mathbf{H}_H}]
\right),
\end{align}
where $\tr_{E_{N+1}}$ is the trace over the final environment space and $\circ_{E_k}$ is the composition over the $k^\text{th}$ environment. 
The state $\rho_{SE,0}^{\mathbf{H}}$ defines the (hypothesis-dependent) initialized state on the system-environment space. Since the map is the same at each time step, so are the adjoining environment spaces. Although the induced multi-time statistics can remain strongly correlated, the process is specified by fixed objects $\rho_{SE,0}^{\mathbf{H}_H}$ (the initial system-environment state) and $H$ (the repeated interaction map). The quantum classical state generated by the action of a $\mathsf{MAD}$ tester on one of the recurrent processes hypotheses $\mathbf{H}_H^N$ with $H\in{\{P,Q\}}$ is denoted by $\rho_{SAC,N}^{\mathsf{T},\mathbf{H}_H}\equiv \rho_{N}^{\mathsf{T},H}$, and marginals are denoted $\hat{\sigma}^{\mathsf{T},\mathbf{H}_H}_{x_{0:N}}\equiv \hat{\sigma}^{{\mathsf{T},H}}_{x_{0:N}}$. 

Recurrent processes arise naturally in settings such as collision models~\cite{ciccarello2022collision}, repeated system–environment interactions~\cite{ciccarello2022collision,AttalPautrat2006,BruneauJoyeMerkli2014}, and quantum memory channels~\cite{KretschmannWerner2005}, where the environment is not reset between uses. In these scenarios, information propagates across time via the environment, leading to non-Markovian temporal correlations~\cite{pollock2018nonmarkovian,Pollock2018,MilzPRX2021Qprocesses}. This makes recurrent processes a natural setting for studying the buildup of distinguishability which we explore in the following sections.

\subsection{Recurrent processes probed by $\mathsf{MAD}$ testers}

For recurrent processes probed by a fixed $\mathsf{MAD}$ tester, each time step applies the same joint operation, so the multi-time evolution is generated by iterating a single map. This structure allows us to upper-bound the distinguishability of processes over time in terms of two components. One that propagates distinguishability accumulated in previous time steps, and another that further increases the distinguishability with each time step.

To make this explicit we define, for each hypothesis $H\in\{P,Q\}$ probed by the $\mathsf{MAD}$ $\mathsf{T}=(\{\mathcal{T}^{(x)}_{SA}\}_x,N)$ such that for each outcome $x\in\mathcal X$, the single-step hypothesis map as
\begin{equation}
\label{eq:Nhx-def}
\mathcal{N}^{\mathsf{T}, H}_x
:=
({H}_{SE}\otimes \mathds{1}_A)\circ(\mathcal{T}^{(x)}_{SA}\otimes \mathds{1}_E),
\end{equation}
so that $\{\mathcal{N}^{\mathsf{T}, H}_x\}_{x\in\mathcal X}$ is a quantum instrument and $\sum_x \mathcal{N}^{\mathsf{T}, H}_x$ is a CPTP map. The full trajectory-conditioned states update with each new observed symbol as
\begin{align}
    \hat{\gamma}^{{\mathsf{T}, H}}_{x_{0:k+1}}
    =
    \mathcal{N}^{\mathsf{T}, H}_{x_{k+1}}[\hat{\gamma}^{{\mathsf{T}, H}}_{x_{0:k}}],
\end{align}
where $\hat{\gamma}_{x_{0:k}}^{\mathsf{T},H}$ is the subnormalized trajectory-conditioned state on the environment, system and quantum ancillary space after observing $x_{0:k}:=(x_0,\dots,x_{k})$. The states $\hat{\gamma}^{{\mathsf{T}, H}}_{x_{0:k}}$ are related to the marginal subnormalized states on the system-ancillary spaces as
\begin{align}
    \hat{\sigma}^{{\mathsf{T}, H}}_{x_{0:k}}
    =
    \tr_{E_{k+1}}(\hat{\gamma}^{{\mathsf{T}, H}}_{x_{0:k}}).
\end{align}
Equivalently, after the next observed symbol,
\begin{align}
    \hat{\sigma}^{{\mathsf{T}, H}}_{x_{0:k+1}}
    =
    \tr_{E_{k+2}}
    \left(
    \mathcal{N}^{\mathsf{T}, H}_{x_{k+1}}[\hat{\gamma}^{{\mathsf{T}, H}}_{x_{0:k}}]
    \right).
\end{align}
We then define the accessible trajectory-wise difference as
\begin{align}
    \Delta_{x_{0:k}}
    :=
    \hat{\sigma}^{{\mathsf{T},P}}_{x_{0:k}}
    -
    \hat{\sigma}^{{\mathsf{T},Q}}_{x_{0:k}},
\end{align}
and the corresponding full trajectory-wise difference as
\begin{align}
    \widetilde{\Delta}_{x_{0:k}}
    :=
    \hat{\gamma}^{{\mathsf{T},P}}_{x_{0:k}}
    -
    \hat{\gamma}^{{\mathsf{T},Q}}_{x_{0:k}}.
\end{align}
Note that, since $\mathcal{N}^{\mathsf{T}, H}_{x_k}$ and $\hat{\gamma}^{{\mathsf{T}, H}}_{x_{0:k}}$ act on the full environment-system-ancillary spaces, the resultant trajectory-conditioned states and trajectory-wise differences will also be functionally dependent on the initial environment-system-ancillary state $\rho_{ESA,0}^{\mathsf{T}, H}=\rho_{ES,0}^{H}\otimes\rho_{A,0}^{\mathsf{T}}$. In particular, the subnormalized states after the first output are given by
\begin{align}
    \hat{\gamma}^{\mathsf{T}, H}_{x_0}
    =
    \mathcal{N}^{\mathsf{T}, H}_{x_0}[\rho_{ESA,0}^{\mathsf{T}, H}].
\end{align}
Recall that from Eq.~\eqref{eq:helstrom-success} and from the classical quantum states in Eq.~\eqref{eq:cq-output}, the optimal probability of successfully distinguishing two recurrent hypotheses $\mathbf{P}_P$ and $\mathbf{Q}_Q$ at time $k$ can be written as
\begin{align}
    p_{succ}^*(\mathbf{P}^k_P,\mathbf{Q}^k_Q,\mathsf{T}^k)
    =
    \frac12(1+\delta_{\mathsf{T}}^k),
\end{align}
where
\begin{align}
    \delta_{\mathsf{T}}^k
    :=
    \frac{1}{2}\sum_{x_{0:k}\in\mathcal{X}^k}
    \big|\!\big|\Delta_{x_{0:k}}\big|\!\big|_1
\end{align}
is the accessible distinguishing bias between the two processes at time $k$ using the $\mathsf{MAD}$ tester $(\mathsf{T},\rho_{A,0}^{\mathsf{T}})$. We also define the full trajectory distinguishability
\begin{align}
    \widetilde{\delta}_{\mathsf{T}}^k
    :=
    \frac{1}{2}\sum_{x_{0:k}\in\mathcal{X}^k}
    \big|\!\big|\widetilde{\Delta}_{x_{0:k}}\big|\!\big|_1.
\end{align}
Since $\hat{\sigma}^{\mathsf{T},H}_{x_{0:k}}$ is obtained from $\hat{\gamma}^{\mathsf{T},H}_{x_{0:k}}$ by tracing out the inaccessible environment, contractivity of the trace norm implies
\begin{align}
    \delta_{\mathsf{T}}^k
    \leq
    \widetilde{\delta}_{\mathsf{T}}^k.
\end{align}

\begin{theorem}\label{thm:stepwise_distinguishability_recurrent}
The accessible distinguishing bias between two recurrent processes $\mathbf{P}_P$ and $\mathbf{Q}_Q$ probed by a $\mathsf{MAD}$ tester $(\mathsf{T},\rho_{A,0}^{\mathsf{T}})$ at time $k+1$ is bounded by the full trajectory distinguishability one step previously, together with a local generation term~$\epsilon_{\mathsf{T}}^k$,
\begin{align}\label{eq:stepwise_distinguishability}
    \delta_{\mathsf{T}}^{k+1}
    \leq
    \widetilde{\delta}_{\mathsf{T}}^k
    +
    \epsilon_{\mathsf{T}}^k,
\end{align}
where
\begin{align}
    \epsilon_{\mathsf{T}}^k
    =
    \frac12
    \sum_{x,x_{0:k}}
    \big|\!\big|
    \tr_{E_{k+2}}
    \big(
    (\mathcal{N}_x^{\mathsf{T},P}-\mathcal{N}_x^{\mathsf{T},Q})
    [\hat{\gamma}_{x_{0:k}}^{\mathsf{T},Q}]
    \big)
    \big|\!\big|_1.
\end{align}
\end{theorem}

\begin{proof}
From the definition of $\Delta_{x_{0:k+1}}$,
\begin{align}
\Delta_{x_{0:k+1}}
=
\tr_{E_{k+2}}
\left(
\mathcal{N}^{\mathsf{T},P}_{x_{k+1}}[\hat{\gamma}^{\mathsf{T},P}_{x_{0:k}}]
-
\mathcal{N}^{\mathsf{T},Q}_{x_{k+1}}[\hat{\gamma}^{\mathsf{T},Q}_{x_{0:k}}]
\right).
\end{align}
Adding and subtracting
$\tr_{E_{k+2}}\left(\mathcal{N}^{\mathsf{T},P}_{x}[\hat{\gamma}^{\mathsf{T},Q}_{x_{0:k}}]\right)$,
taking the trace norm, applying the triangle inequality, and summing over all $x_{0:k+1}\equiv (x_{0:k},x_{k+1}) = (x_{0:k},x)  $, we obtain
\begin{align}
\sum_{x_{0:k+1}}
\big|\!\big|\Delta_{x_{0:k+1}}\big|\!\big|_1
\leq
&
\sum_{x_{0:k},x}
\big|\!\big|
\tr_{E_{k+2}}
\left(
\mathcal{N}^{\mathsf{T},P}_{x}
[\widetilde{\Delta}_{x_{0:k}}]
\right)
\big|\!\big|_1
\nonumber\\
&
+
\sum_{x_{0:k},x}
\big|\!\big|
\tr_{E_{k+2}}
\left(
(\mathcal{N}^{\mathsf{T},P}_{x}-\mathcal{N}^{\mathsf{T},Q}_{x})
[\hat{\gamma}^{\mathsf{T},Q}_{x_{0:k}}]
\right)
\big|\!\big|_1.
\end{align}
The first term is bounded by the previous full trajectory distinguishability, since $\{\mathcal{N}^{\mathsf{T},P}_{x}\}_x$ is a quantum instrument and the partial trace is CPTP. Therefore
\begin{align}
\sum_{x_{0:k},x}
\big|\!\big|
\tr_{E_{k+2}}
\left(
\mathcal{N}^{\mathsf{T},P}_{x}
[\widetilde{\Delta}_{x_{0:k}}]
\right)
\big|\!\big|_1
\leq
\sum_{x_{0:k}}
\big|\!\big|
\widetilde{\Delta}_{x_{0:k}}
\big|\!\big|_1.
\end{align}
Dividing both sides by $2$ yields Eq.~\eqref{eq:stepwise_distinguishability} as claimed. A detailed proof is given in Appendix.~\ref{app:proof_of_thm_stepwise}.
\end{proof}

Theorem~\ref{thm:stepwise_distinguishability_recurrent} shows that the evolution of the accessible distinguishing bias can be upper bounded by two contributions. The first is the propagation of distinguishability stored in the full trajectory-conditioned state, including distinguishability that may be inaccessible at time $k$ because it is stored in the environment. The second is the local distinguishability $\epsilon_{\mathsf{T}}^k$ generated between the two hypotheses along the induced trajectories at each time step. This decomposition demonstrates that the performance of a $\mathsf{MAD}$ tester is governed by the interplay between how efficiently distinguishability is generated, how much of it is stored in inaccessible degrees of freedom, and how well it is preserved under the induced dynamics.

\subsection{Heuristic behavior of $\mathsf{MAD}$ distinguishability for recurrent processes}\label{sec:heuristics}

We can develop further insight into the evolution of distinguishability in this regime by noting that, from a physical perspective, we must also account for the loss of memory induced by the dynamics. In particular, if the dynamics generated by the instrument $\{\mathcal{N}_x\}_x$ is mixing, then deviations will be progressively suppressed~\cite{Wolf2012,Kossakowski2012,Souissi2025exponential} under repeated application, so that distinguishability accumulated at earlier times is effectively damped. While the rigorous bound above involves the full trajectory distinguishability $\widetilde{\delta}_{\mathsf T}^k$, the operational success probability is determined by the accessible distinguishability $\delta_{\mathsf T}^k$. The difference between these quantities represents distinguishability that is temporarily stored in the inaccessible environment. Under mixing dynamics, this hidden contribution is not expected to remain coherently available indefinitely, and its effect on the accessible trajectory distinguishability is captured only after coarse graining over trajectories and tracing out the environment.

Motivated by these observations, we consider a heuristic description of the accessible dynamics,
\begin{align}
    \delta_{\mathsf T}^{k+1}
    \approx
    \lambda\, \delta_{\mathsf T}^k
    +
    \epsilon_{\mathsf T}^k,
\end{align}
where the first term captures effective propagation with loss of memory, and the second term captures new distinguishability together with any short-lived environmental contribution that becomes accessible at the next step. Here $\lambda \in (0,1)$ should be understood as an effective propagation parameter, rather than as a microscopic contraction coefficient acting on every trajectory. It describes the net fraction of accessible distinguishability that survives one further recurrent step after averaging over trajectories and tracing out the inaccessible environment. Such behavior is known in the study of quantum Markov processes, where exponential mixing is governed by the spectral gap of the channel~\cite{Kossakowski2012,Souissi2025exponential}, and more generally by contraction coefficients that quantify the degradation of distinguishability under noisy evolution~\cite{Temme2010,Hirche2022}.

We may also consider that the distinguishability that is added per time step becomes approximately constant along most trajectories so that $\epsilon_{\mathsf T}^k \approx m > 0$. This is physically reasonable when the induced dynamics is mixing, as the trajectory-conditioned states $\hat{\sigma}_{x_{0:k}}^{\mathsf{T},H}$ converge towards a stationary ensemble~\cite{Kossakowski2012,Souissi2025exponential,BruneauJoyeMerkli2014}, and local observables, including the distinguishability between $\mathcal{N}_x^{\mathsf{T},P}$ and $\mathcal{N}_x^{\mathsf{T},Q}$ evaluated on these states, become approximately time-independent~\cite{Temme2010,Souissi2025exponential,BruneauJoyeMerkli2014} up to small fluctuations. Then, the dynamics reduces to a driven linear recursion
\begin{align}
    \delta_{\mathsf T}^{k+1}
    \approx
    \lambda\, \delta_{\mathsf T}^k
    +
    m,
\end{align}
with solution given by
\begin{align}\label{eq:recursion_solution}
    \delta_{\mathsf T}^N
    \approx
    \delta_\infty
    +
    (\delta_{\mathsf T}^0 - \delta_\infty)\lambda^N,
    \qquad
    \delta_\infty
    =
    \frac{m}{1-\lambda}.
\end{align}
Equation~\eqref{eq:recursion_solution} essentially describes exponential relaxation towards a steady distinguishability. Writing $\lambda^N = e^{-cN}$ with $c = -\log \lambda > 0$, this becomes
\begin{align}
    \delta_{\mathsf T}^N
    \approx
    \delta_\infty
    -
    A e^{-cN},
\end{align}
where $A$ absorbs initial conditions and short-time transients.

Since $p_{\mathrm{succ}}^\ast = \frac12(1 + \delta_{\mathsf T}^N)$, this immediately implies exponential saturation of the success probability,
\begin{align}\label{eqn:p_success_heuristic}
\begin{split}
        p_{\mathrm{succ}}^\ast(N)
    &\approx
    p_\infty
    -
    \frac12 A e^{-cN},
    \\
    p_\infty
    &=
    \frac12(1 + \delta_\infty).
\end{split}
\end{align}
In general, $p_\infty < 1$, reflecting the fact that a fixed, finite-memory tester may not extract all available distinguishability between the processes. In the unconstrained case, where the coherent memory dimension may grow with the number of time steps, the finite-time optimum is recovered by the strategy-norm distinguishability~\cite{GutoskiWatrous2007,gutoski2012measure}. When recurrent dynamics continually generates distinguishability that is propagated through the induced multi-time statistics, as in the repeated-interaction models considered below, the strategy-norm success probability is expected to approach one as $N$ increases. This need not occur if the hypotheses differ only through transient dynamics, for example through their initial system-environment states, since such differences may be washed out under mixing dynamics~\cite{Kossakowski2012,Souissi2025exponential,BruneauJoyeMerkli2014}, or if their difference lies in degrees of freedom that are never reached, or reached only with vanishing weight, under admissible interventions. In these cases the induced multi-time statistics may provide only bounded evidence, so even the unconstrained success probability may saturate below one.

This picture can be understood physically as a driven-dissipative process~\cite{BreuerPetruccione2002,ciccarello2022collision}. Distinguishability is continuously added at each step, while contributions generated in the past are progressively erased by the mixing dynamics. As a result, although the $\mathsf{MAD}$ tester stores an ever-growing classical record $x_{0:N}$, the additional distinguishability available for extraction decreases over time, so that later measurements yield progressively less information. Consequently, the total distinguishability saturates despite the unbounded growth of the record.

We now test this heuristic prediction numerically in a concrete model.

\section{Numerical illustration: $\mathsf{MAD}$ distinguishability as an efficiently calculable alternative to the strategy-norm distance}
\label{sec:numerics}

\begin{figure}
    \centering
    \includegraphics[width=\linewidth]{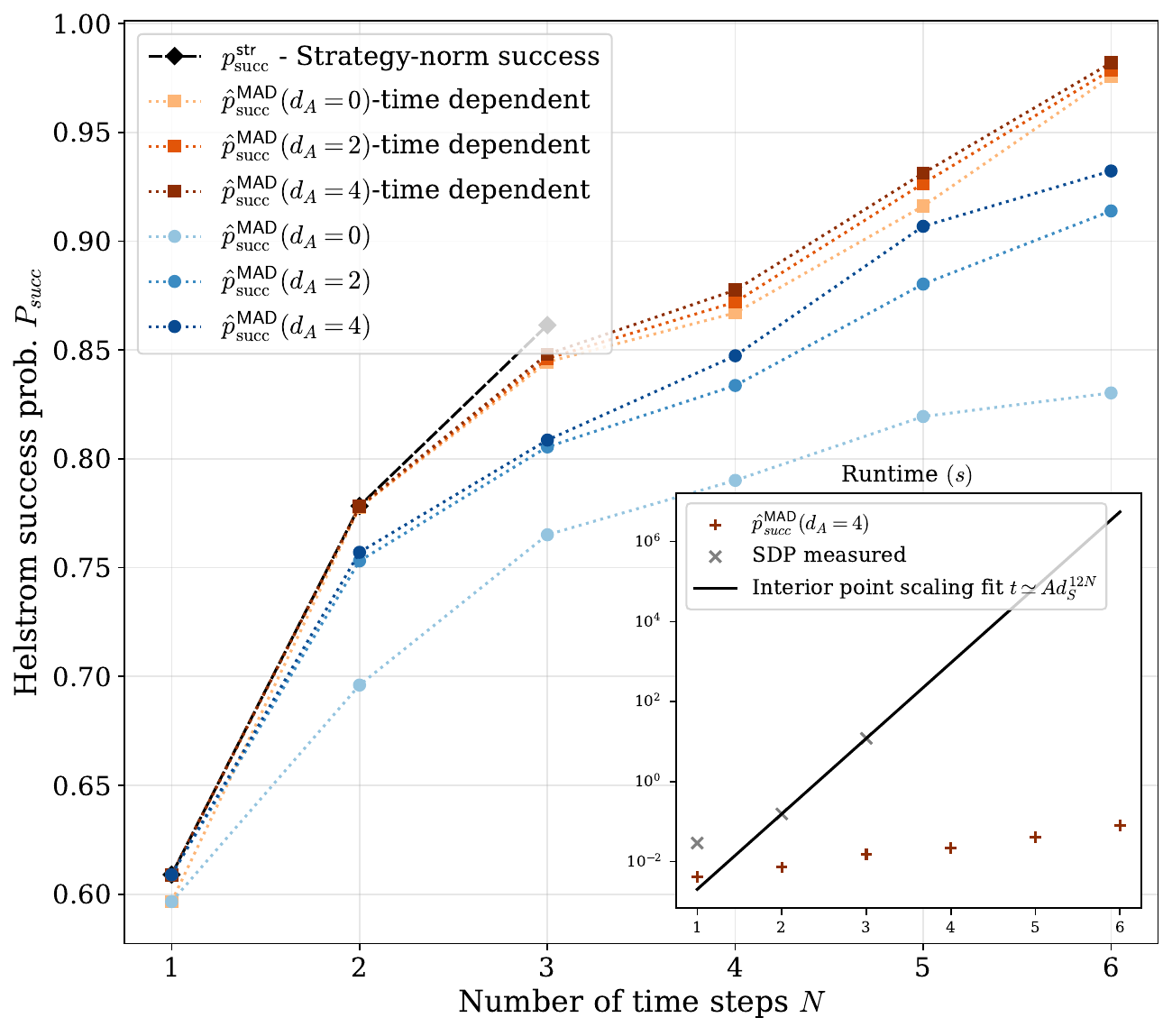}
    \caption{Helstrom success probability $p_{\mathrm{succ}}^{\mathsf{MAD}}$ for numerically optimized $\mathsf{MAD}$ testers distinguishing the repeated-interaction processes generated by the unitaries in Eq.~\eqref{eq: numerical_process_unitaries}, with $\theta_{\mathcal P}=0.2$ and $\theta_{\mathcal Q}=0.5$, compared with the strategy-norm success benchmark. Increasing the coherent ancilla dimension $d_A\in\{0,2,4\}$, for a time-dependent strategy (orange lines), or time-independent strategy (blue lines), systematically improves performance, illustrating the coherent-memory $\mathsf{MAD}$ hierarchy and its convergence toward the strategy-norm benchmark. For fixed coherent ancilla dimension $d_A=0$ (no coherent memory), allowing explicit time dependence through a classical counter substantially improves the success probability relative to time-independent recurrent $\mathsf{MAD}$ testers. The inset compares the measured runtime of our SDP implementation for the values $N = 1,2,3$, where the full strategy-norm optimisation could be performed. The black curve in the inset is a fit of the expected interior-point scaling $t_{\text{SDP}} \simeq Ad_S^{12N}$ to these SDP runtime data, and is included only as a visual guide to the expected asymptotic growth of dense SDP methods. In contrast, the measured runtime of the MAD search exhibits a much milder increase over the range of N considered.}
    \label{fig:mad_clock_numerics}
\end{figure}

We turn to a numerical illustration of the framework developed above, benchmarking $\mathsf{MAD}$ testers against the strategy-norm optimum and examining how memory constraints and dynamical mixing govern the buildup of distinguishability. We also show that increasing the available memory improves distinguishing performance, while remaining significantly more efficient to evaluate than the full strategy-norm SDP. To this end, we consider a repeated-interaction model and define two process hypotheses $P$ and $Q$ via different joint system--environment unitaries applied at each time step.

Writing $\rho_{SE,k}$ for the joint system-environment state after $k$ interactions,
\begin{align}
\rho_{SE,k+1}=U_X\,\rho_{SE,k}\,U_X^\dagger,
\qquad X\in\{P,Q\},
\end{align}
with the same environment carried forward between steps. Concretely, we take partial-SWAP dynamics,
\begin{align}\label{eq: numerical_process_unitaries}
\begin{split}
    U_{P}&=\exp\!\big[\!\!-i\theta_{P}\,(\mathrm{SWAP})_{SE}\big];~
\\U_{Q}&=\exp\!\big[\!\!-i\theta_{Q}\,(\mathrm{SWAP})_{SE}\big],
\end{split}
\end{align}
with $\theta_{P}\neq \theta_{Q}$. The resulting $N$-step processes $\mathbf{P}^{N}$ and $\mathbf{Q}^{N}$ are recurrent and memoryful, with temporal correlations mediated by the persistent environment. For each time $N$ we construct the process Choi operators $\Upsilon^{\mathbf{P}^{N}}$ and~$\Upsilon^{\mathbf{Q}^{N}}$, and compute the strategy-norm distinguishability $d_{\mathrm{str}}^{(N)}(\mathbf{P}^{N}_P,\mathbf{Q}^{N}_Q)$ using the SDP formulation (see Appendices.~\ref{app:combs} and ~\ref{app:strategy-norm-sdp} for details). This yields a benchmark success probability
\begin{align}
    p^{\mathrm{str}*}_{succ}(\mathbf{P}^{N},\mathbf{Q}^{N})=\frac12\left(1+d_{\mathrm{str}}^{(N)}(\mathbf{P}^{N}\!,\!\mathbf{Q}^{N})\right)
\end{align}

For an $N$-step process, the corresponding SDP required to evaluate $p^{\mathrm{str}*}_{succ}$ involves optimizing a positive semidefinite matrix of size $n = d^{2N}$ together with causality constraints whose number scales as $m \sim d^{4N}$. For state-of-the-art interior-point methods, the per-iteration cost scales as $\mathcal{O}(m n^2 + m^2 n + m^3)$ \cite{Jiang2020AFI}. This leads to a runtime
that grows as $d_S^{12N}$ per iteration~\cite{Watrous2009SDP, mironowicz2024sdp}, and renders direct optimization infeasible beyond small $N$. This is reflected in Fig.~\ref{fig:mad_clock_numerics} where we have calculated $p^{\mathrm{str}*}_{succ}$ only to $N=3$, and where the inset shows the growing runtime of the SDP (the gray markers show the measured times of our SDP implementation for the values of $N$ for which the full strategy-norm optimization could be performed). The `interior-point scaling fit' curve is obtained by fitting the expected dense interior-point scaling $t_{\mathrm{SDP}}(N) \simeq A d_S^{12N}$ to the measured runtimes and serves as a visual guide to the anticipated asymptotic growth for state of the art interior-point SDP methods~\cite{Jiang2020AFI} and illustrates why direct evaluation of $p_{\mathrm{succ}}^{\mathrm{str},*}$ rapidly becomes computationally prohibitive as $N$ increases.

Motivated by this computational bottleneck, we evaluate bounded-memory $\mathsf{MAD}$ testers using a direct numerical search. For each coherent ancilla dimension $d_A$, the instrument is parameterised by a unitary dilation acting on the system and ancilla, and candidate testers are obtained by sampling Haar-random unitaries and applying a local optimisation routine in Python. The resulting classical--quantum output states are then distinguished using the Helstrom measurement. Since the optimisation over $\mathsf{MAD}$ testers is non-convex, the plotted values should be interpreted as the best values found by this search procedure, rather than as certified global optima over all testers in $\mathsf{MAD}(d_A)$. Nevertheless, every tester generated in this way is admissible, so the numerical values provide rigorous lower bounds on the true memory-constrained optima, 

\begin{align}
   \begin{split}
       {p}_{\mathrm{succ}}^{\mathsf{MAD}}(\mathbf{P}^N\!,\mathbf{Q}^N\!;d_A) &\leq p_{\mathrm{succ}}^{\mathsf{MAD},*}(\mathbf{P}^N\!,\mathbf{Q}^N\!;d_A) \\&\leq p_{\mathrm{succ}}^{\mathrm{str}*}(\mathbf{P}^N\!,\mathbf{Q}^N)
   \end{split}
\end{align} 
where ${p}_{\mathrm{succ}}^{\mathsf{MAD}}(\mathbf{P}^N\!,\mathbf{Q}^N;d_A)$ is the best found value, $p_{\mathrm{succ}}^{\mathsf{MAD},*}(\mathbf{P}^N\!,\mathbf{Q}^N;d_A)$ is the true $\mathsf{MAD}$ optimal value and $p_{\mathrm{succ}}^{\mathrm{str}*}(\mathbf{P}^N\!,\mathbf{Q}^N)$ is the true strategy-norm distance optimal probability.
The agreement with the strategy-norm benchmark at small $N$, where the latter is available, indicates that the search already finds near-optimal testers for the examples considered here.

Given a $\mathsf{MAD}$ tester, the corresponding success probability is computed from the resulting classical-quantum output state, taking into account both the full classical outcome log and the final coherent
memory. Specifically, we evaluate
\begin{equation}
p^{\mathsf{MAD}*}_{succ}(\mathbf{P}^{N}\!,\mathbf{Q}^{N}\!;d_A)
=
\tfrac12\!\left(1+d_{\mathsf{MAD}}^{(N)}(\mathbf{P}^{N}\!,\mathbf{Q}^{N}\!;d_A)\right),
\end{equation}
where $d_{\mathrm{MAD}}^{(N)}$ is defined in Eq.~\eqref{eq:drec-def}. We compare two cases that match the $\mathsf{MAD}$ framework. First, a $\mathsf{MAD}$ instrument without internal time-keeping and second, a counter-controlled $\mathsf{MAD}$ instrument in which an internal counter routes outcomes and selects the appropriate branch at each use (Fig.~\ref{fig:MAD}). This comparison isolates the role of classical time control in enhancing distinguishability at fixed quantum memory. 

Figure~\ref{fig:mad_clock_numerics} shows the $\mathsf{MAD}$ memory hierarchy. Increasing the coherent ancilla dimension $d_A$ systematically improves $p_{\mathrm{succ}}^{\mathrm{MAD}}(N;d_A)$, both for time-independent testers (blue lines) and for counter-routed, time-dependent testers (orange lines), and brings it closer to the strategy-norm benchmark. This behavior directly reflects monotonicity in coherent memory and the saturation of the hierarchy at large $d_A$. That is, larger coherent memory allows the tester to preserve more time-nonlocal information about the interaction history and to implement more powerful internal feedback.

The counter-routed curves also show that explicit time dependence can substantially improve discrimination. This is expected because the counter allows the effective intervention implemented by the tester, encoded by the family of instruments $\{\mathcal{T}_{x_k}\}_{x_k,k}$ to depend on the time index $k$, while the overall device remains recurrent as in Theorem~\ref{thm:mad-compilation}. However, this improvement comes from an additional control resource, since one must be able to implement, or variationally optimize over, a different effective quantum instrument at different process slots. From a computational perspective, this also leads to a less scalable optimization problem, as the number of independently tunable operations grows with the number of process slots, in contrast to the recurrent time-independent setting where the same instrument is reused throughout. Thus, the time-dependent curves should be viewed as a more powerful comparison class of testers rather than as the primary operational restriction considered here. 

The central resource isolated by the $\mathsf{MAD}$ hierarchy is instead the coherent memory dimension $d_A$. By fixing the recurrent structure of the tester and varying $d_A$, one directly probes how much information about the past interaction history must be retained coherently in order to distinguish the processes. The numerical results make this dependence explicit. Increasing $d_A$ yields consistent gains in discrimination performance within both the time-independent and time-dependent classes, while the hierarchy approaches the unrestricted strategy-norm benchmark as $d_A$ becomes large. This behavior is precisely what one would expect if coherent memory is the limiting resource governing multi-time discrimination. In this sense, the figure illustrates two complementary mechanisms for improving performance. Explicit time dependence provides additional control over the probing strategy, whereas increasing $d_A$ expands the amount of temporal information that can be stored and processed coherently. The latter is the resource quantified by the $\mathsf{MAD}$ hierarchy, and the observed monotonic improvement with $d_A$ provides direct operational evidence that coherent memory controls the power of recurrent testers in distinguishing quantum processes.

Figure~\ref{fig:Helstrom_success_with_fit} provides strong numerical support for the effective description developed in Sec.~\ref{sec:heuristics}. In particular, the results support the idea that driven--dissipative dynamics govern the distinguishability which is continuously generated at each step, while also being suppressed by mixing. Here, the simulated success probability is well described by the form $p_{\mathrm{succ}}^\ast(N) \approx p_\infty - A e^{-cN}$~\footnote{Fit values for $d_A=(0,4)$ are, respectively, $p_\infty=(0.86,0.89)$, $A=(0.83,0.89)$ and $c=(0.45,0.46)$.}, demonstrating that the distinguishability dynamics exhibits exponential relaxation towards a steady value.

This behavior is consistent with the recursion $\delta_{\mathsf T}^{k+1} \!\approx\!\lambda \delta_{\mathsf T}^k\! +\! m$ and indicates that, despite the underlying trajectory-dependent dynamics, the ensemble-averaged evolution is effectively governed by a single contraction rate and a constant addition of distinguishability. Moreover, the asymptotic value $p_\infty < 1$ reflects a genuine limitation of finite-memory $\mathsf{MAD}$ testers. Even at long times, the tester does not extract all available distinguishability between the processes despite significantly increasing the initial $p_{\text{succ}}$ of around $60\%$ to above $85\%$ for $N=10$ time steps. This is further supported by the dependence of $p_\infty$ on the coherent memory dimension $d_A$, with larger memory enabling higher asymptotic success probability. In other words, the saturation value encodes the interplay between dynamical mixing and the memory constraints of the tester.

\begin{figure}
    \centering
    \includegraphics[width=0.95\linewidth]{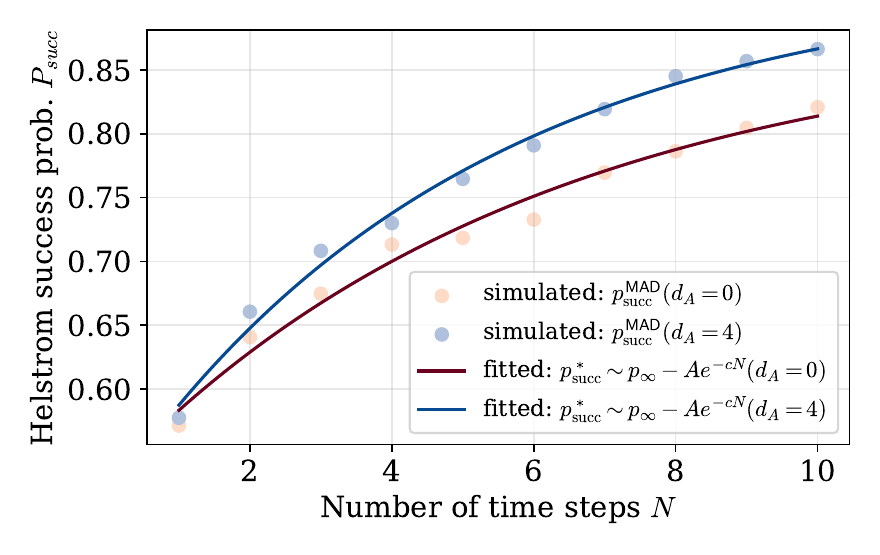}
    \caption{Helstrom success probability $P_{\mathrm{succ}}$ for numerically optimized fixed recurrent $\mathsf{MAD}$ testers versus the number of time steps $N$, for coherent ancilla dimensions $d_A=0$ (no coherent memory) and $d_A=4$. The pink and blue markers show the simulated values for $d_A=0$ and $d_A=4$, respectively. The corresponding solid maroon and blue curves are fits to the heuristic model $p_{\mathrm{succ}}^\ast \sim p_\infty - A e^{-cN}$. The good agreement indicates that the distinguishability dynamics is well captured by an effective driven-dissipative picture, with exponential approach to a fixed value.
}
    \label{fig:Helstrom_success_with_fit}
\end{figure}

\section{Conclusion and outlook}\label{sec:conclusion}

We introduced machines for autonomous distinction ($\mathsf{MAD}$s) as a framework for distinguishing multi-time quantum processes using testers with bounded coherent memory. The quantity
$d^{(N)}_{\mathrm{MAD}}(\mathbf{P}^N,\mathbf{Q}^N;d_A)$ gives the largest distinguishing bias achievable when the experimenter repeatedly probes the process with a fixed instrument, retains the full classical outcome record, and carries a coherent memory of dimension at most $d_A$. In this way, process distinguishability becomes an operational measure of accessible temporal information, quantifying which differences between processes are visible to an agent with limited quantum memory \cite{chiribella2009theoretical,pollock2018nonmarkovian,MilzPRX2021Qprocesses,ZambonDistinguishability2024}.

The resulting hierarchy interpolates between experimentally restricted and fully adaptive discrimination. For fixed $d_A$, $\mathsf{MAD}$ distinguishability gives an operationally meaningful lower bound on the strategy-norm distance. In the large-memory limit, however, the restriction is complete for finite times. Any admissible $N$-step tester can be compiled into a $\mathsf{MAD}$ tester with an internal counter and sufficiently large coherent memory, so that the hierarchy saturates the strategy-norm benchmark \cite{GutoskiWatrous2007,gutoski2012measure,watrous2018theory}. This identifies coherent memory as a key resource controlling the gap between realistic recurrent probing and fully adaptive process discrimination, in line with recent work on memory as a resource in multi-time quantum processes and channel discrimination \cite{milz2020,taranto2021,giarmatzi2023,taranto2024hierarchy,kechrimparis2025quantumagentscomplexity,ohst2026}. Indeed, in quantum channel discrimination, adaptive strategies can be strictly more powerful than non-adaptive ones, and refined classes of probing strategies, such as parallel, sequential, and indefinite-causal-order strategies, can form strict hierarchies of discrimination power~\cite{harrow2010adaptive,bavaresco2021strict}. These results concern memory and adaptivity on the probing side; here, we ask how such resources interact with temporal memory in the process itself.

For recurrent processes generated by repeated system--environment interactions, we derived a single-step description of the induced distinguishability. This separates the generation of new distinguishability from the propagation and decay of distinguishability generated at earlier times. The numerical examples show that increasing $d_A$ systematically improves the $\mathsf{MAD}$ success probability and closes the gap to the strategy-norm benchmark, while avoiding the full semidefinite program (SDP) over general testers. The latter has runtime scaling of order $\mathcal{O}(d_S^{12N})$ with state-of-the-art interior-point methods for solving the SDP, where $d_S$ is the system dimension~\cite{Jiang2020AFI}. This makes the framework particularly useful when the relevant temporal correlations can be compressed into a small coherent memory, as is often expected for mixing dynamics, finite-memory environments, collision models, or process tensors admitting compact tensor-network representations \cite{ciccarello2022collision,KretschmannWerner2005,jorgensen2020discrete,cygorek2022compression,dowling2024capturing}.

The same perspective also clarifies when finite-memory $\mathsf{MAD}$ distinguishability should perform poorly. Since a $\mathsf{MAD}$ only retains a bounded coherent memory, it is not expected to approximate the strategy-norm distance well when the relevant distinguishing information is genuinely quantum \cite{milz2020,taranto2024hierarchy}. This can occur when process differences are distributed coherently across many time steps, especially in processes with long-range temporal correlations or memory structures that are not efficiently captured by a small sequential memory~\cite{jorgensen2020discrete,dowling2024capturing}. It can also occur when information is stored in environmental degrees of freedom for long periods before being returned to the system, as in repeated-interaction and collision-model settings with persistent environments \cite{KretschmannWerner2005,ciccarello2022collision}. More generally, slowly mixing, non-ergodic, near-critical, or otherwise highly structured environments are natural regimes in which the effective memory timescale of the process may exceed the coherent memory available to the tester \cite{fux2021efficient,cygorek2022compression,ivander2024unified}. In such cases, convergence of the $\mathsf{MAD}$ hierarchy towards the strategy norm may require substantially larger ancilla dimensions. The gap between $d^{(N)}_{\mathrm{MAD}}$ and the strategy norm can therefore be interpreted not merely as a limitation of the restricted tester class, but as a diagnostic of the coherent temporal memory required to discriminate the underlying processes.

Several future directions and use cases follow from this viewpoint. First, the $\mathsf{MAD}$ hierarchy suggests a variational, tensor-network-inspired approach to process discrimination, with the coherent memory dimension $d_A$ playing a role analogous to a bond dimension \cite{Orus2014}. Rather than optimizing over all admissible testers, one can optimize over a family of recurrent testers of increasing coherent memory dimension. This gives a controlled sequence of lower bounds on the strategy-norm distinguishability and connects naturally with process-tensor methods for simulating, learning, and compressing non-Markovian dynamics \cite{fux2024oqupy,cygorek2024ace,keeling2025processTensorApproaches}.

Second, $\mathsf{MAD}$ distinguishability provides a diagnostic for process-tensor learning, benchmarking, and noise characterization. This complements non-Markovian quantum process tomography and recent efficient frameworks for non-Markovian characterization~\cite{whiteNonMarkovian2022PRX,white2025unifyingPRX}. Rather than reconstructing a full multi-time process, one may ask which process differences are visible to a finite-memory probing agent. This is relevant for experiments in which temporally correlated noise, crosstalk, or calibration drift must be detected with restricted control and limited coherent storage~\cite{white2020demonstration,Zhang2022NonMarkovianSuperconducting,WhitePRL2023,Tripathi2024BenchmarkingGates,Hashim2025QCVV,zhang2025learningForecasting}. In this setting, the $\mathsf{MAD}$ success probability gives an operational benchmark for the amount of non-Markovian information that can actually be extracted by realistic probing.

Third, the dependence of $d^{(N)}_{\mathsf{MAD}}$ on $d_A$ can be used as a memory diagnostic. Rapid convergence of the hierarchy indicates that the process differences relevant for discrimination are compressible into a small coherent memory. By contrast, a persistent gap between $d^{(N)}_{\mathsf{MAD}}$ and the strategy-norm distance signals that the processes differ in temporal correlations that require larger coherent memory to access. This makes the hierarchy useful not only as an approximation scheme, but also as a way of estimating the coherent memory resources required to detect non-Markovian effects and long-range temporal structure~\cite{milz2020,taranto2021,taranto2024hierarchy,ohst2026}.

Fourth, the recurrent structure of $\mathsf{MAD}$ testers connects directly to quantum reservoir computing and recurrent quantum models. A reservoir is useful precisely when it maps temporal inputs into distinguishable internal states and readout statistics. The $\mathsf{MAD}$ framework provides a discrimination-based way to quantify this ability, and may therefore serve as a benchmark for input separation, memory capacity, fading-memory behavior, and dissipative information processing in quantum reservoirs \cite{Ghosh2019_Reservoir_npj,fujii2021quantum,nakajima2019boosting,mujal2021opportunities,govia2021reservoir,suzuki2022natural,sannia2024dissipation,martinezpena2025inputDependence,wringe2025reservoirBenchmarks}.

Further extensions include treating the available control structure as a resource distinct from coherent memory. In the present framework, time dependence and feedback can be encoded through an internal counter, allowing the $\mathsf{MAD}$ hierarchy to isolate coherent memory as the constrained resource. In realistic implementations, however, the ability to vary the probing instrument over time, condition operations on previous outcomes, or implement a large family of controlled interventions may itself be limited. A refined hierarchy could therefore constrain both coherent memory and control structure, interpolating between strictly time-homogeneous recurrent testers, counter-controlled $\mathsf{MAD}$s, and fully adaptive comb testers \cite{chiribella2009theoretical,GutoskiWatrous2007,gutoski2012measure,watrous2018theory}. This would help separate limitations caused by insufficient coherent memory from those caused by restricted control, connecting the present framework to broader questions about resources, costs, and complexity in quantum control \cite{milz2020,taranto2024hierarchy}.

Taken together, these directions suggest that bounded-memory distinguishability is not only a tractable approximation to an ideal benchmark, but also a way of probing the operational structure of temporal quantum processes. The $\mathsf{MAD}$ hierarchy identifies which process differences are visible to agents with limited coherent memory, which differences require more powerful temporal resources, and how those requirements change with the structure of the process. In this sense, $\mathsf{MAD}$ distinguishability provides a bridge between the abstract strategy-norm theory of quantum processes and the finite-memory setups available in experiment, simulation, and recurrent quantum information processing.

\section{Acknowledgments}
The authors would like to thank Simon Milz, Alexander B. Boyd, and Marco Radaelli for insightful comments. The research conducted in this publication was funded by Taighde Éireann -- Research Ireland under grant number IRCLA/2022/3922.
\bibliography{mybib}

@article{ZambonDistinguishability2024,
  title = {Process tensor distinguishability measures},
  author = {Zambon, Guilherme},
  journal = {Phys. Rev. A},
  volume = {110},
  issue = {4},
  pages = {042210},
  numpages = {9},
  year = {2024},
  month = {Oct},
  publisher = {American Physical Society},
  doi = {10.1103/PhysRevA.110.042210},
  url = {https://link.aps.org/doi/10.1103/PhysRevA.110.042210}
}

@misc{Watrous2009SDP,
      title={Semidefinite programs for completely bounded norms}, 
      author={John Watrous},
      year={2009},
      primaryClass={quant-ph},
      url={https://arxiv.org/abs/0901.4709}, 
}

@article{MilzPRX2021Qprocesses,
  title = {Quantum Stochastic Processes and Quantum non-Markovian Phenomena},
  author = {Milz, Simon and Modi, Kavan},
  journal = {PRX Quantum},
  volume = {2},
  issue = {3},
  pages = {030201},
  numpages = {81},
  year = {2021},
  month = {Jul},
  publisher = {American Physical Society},
  doi = {10.1103/PRXQuantum.2.030201},
  url = {https://link.aps.org/doi/10.1103/PRXQuantum.2.030201}
}

@article{Berk2021ResourceTheoriesOf,
  author       = {Berk, Graeme D. and Garner, Andrew J. P. and Yadin, Benjamin and Modi, Kavan and Pollock, Felix A.},
  title        = {Resource theories of multi-time processes: A window into quantum non-Markovianity},
  journal      = {Quantum},
  volume       = {5},
  pages        = {435},
  year         = {2021},
  doi          = {10.22331/q-2021-04-20-435},
}

@article{Distinguishability_Yang_2020_PRE,
  title = {Measures of distinguishability between stochastic processes},
  author = {Yang, Chengran and Binder, Felix C. and Gu, Mile and Elliott, Thomas J.},
  journal = {Phys. Rev. E},
  volume = {101},
  issue = {6},
  pages = {062137},
  numpages = {10},
  year = {2020},
  month = {Jun},
  publisher = {American Physical Society},
  doi = {10.1103/PhysRevE.101.062137},
  url = {https://link.aps.org/doi/10.1103/PhysRevE.101.062137}
}

@article{Ghosh2019_Reservoir_npj,
  title = {Quantum Reservoir Processing},
  author = {Ghosh, Sarma and Opanchuk, Bogdan and Rosales-Zárate, Laura and Wilson, Bela and Reid, Margaret D.},
  journal = {npj Quantum Information},
  volume = {5},
  number = {1},
  pages = {35},
  year = {2019},
  doi = {10.1038/s41534-019-0148-8}
}

@article{GilchristLangfordNielsen2005,
  author       = {A. Gilchrist and N. K. Langford and M. A. Nielsen},
  title        = {Distance measures to compare real and ideal quantum processes},
  journal      = {Phys. Rev. A},
  volume       = {71},
  pages        = {062310},
  year         = {2005},
  doi          = {10.1103/PhysRevA.71.062310}
}

@article{PollockRodriguezRosarioFrauenheimPaternostroModi2018,
  author       = {F. A. Pollock and C. A. Rodr{\'i}guez-Rosario and T. Frauenheim and M. Paternostro and K. Modi},
  title        = {Non-Markovian quantum processes: Complete framework and efficient characterization},
  journal      = {Phys. Rev. A},
  volume       = {97},
  pages        = {012127},
  year         = {2018},
  doi          = {10.1103/PhysRevA.97.012127}
}

@article{Pollock2018,
  author       = {F. A. Pollock and C. A. Rodr{\'i}guez-Rosario and T. Frauenheim and M. Paternostro and K. Modi},
  title        = {Operational Markov Condition for Quantum Processes},
  journal      = {Phys. Rev. Lett.},
  volume       = {120},
  pages        = {040405},
  year         = {2018},
  doi          = {10.1103/PhysRevLett.120.040405}
}

@article{MilzPollockModi2016,
  author       = {S. Milz and F. A. Pollock and K. Modi},
  title        = {Reconstructing open quantum system dynamics with limited control},
  journal      = {Phys. Rev. A},
  volume       = {96},
  pages        = {042109},
  year         = {2017},
  doi          = {10.1103/PhysRevA.96.042109}
}

@article{Tripathi2024BenchmarkingGates,
  author  = {Vinay Tripathi and Daria Kowsari and Kumar Saurav and Haimeng Zhang and Eli M. Levenson-Falk and Daniel A. Lidar},
  title   = {Benchmarking Quantum Gates and Circuits},
  journal = {Chemical Reviews},
  volume  = {125},
  number  = {12},
  pages   = {5745--5775},
  year    = {2024},
  doi     = {10.1021/acs.chemrev.4c00870}
}

@article{Hashim2025QCVV,
  title = {Practical Introduction to Benchmarking and Characterization of Quantum Computers},
  author = {Hashim, Akel and Nguyen, Long B. and Goss, Noah and Marinelli, Brian and Naik, Ravi K. and Chistolini, Trevor and Hines, Jordan and Marceaux, J.P. and Kim, Yosep and Gokhale, Pranav and Tomesh, Teague and Chen, Senrui and Jiang, Liang and Ferracin, Samuele and Rudinger, Kenneth and Proctor, Timothy and Young, Kevin C. and Siddiqi, Irfan and Blume-Kohout, Robin},
  journal = {PRX Quantum},
  volume = {6},
  issue = {3},
  pages = {030202},
  numpages = {132},
  year = {2025},
  month = {Aug},
  publisher = {American Physical Society},
  doi = {10.1103/PRXQuantum.6.030202},
  url = {https://link.aps.org/doi/10.1103/PRXQuantum.6.030202}
}

@article{NakajimaFujiiNegoroMitaraiKitagawa2018,
  author       = {K. Nakajima and K. Fujii and M. Negoro and K. Mitarai and M. Kitagawa},
  title        = {Boosting computational power through spatial multiplexing in quantum reservoir computing},
  journal      = {Phys. Rev. Applied},
  volume       = {11},
  pages        = {034021},
  year         = {2019},
  doi          = {10.1103/PhysRevApplied.11.034021}
}

@article{CiracVerstraete2021,
  title = {Matrix product states and projected entangled pair states: Concepts, symmetries, theorems},
  author = {Cirac, J. Ignacio and P\'erez-Garc\'{\i}a, David and Schuch, Norbert and Verstraete, Frank},
  journal = {Rev. Mod. Phys.},
  volume = {93},
  issue = {4},
  pages = {045003},
  numpages = {65},
  year = {2021},
  month = {Dec},
  publisher = {American Physical Society},
  doi = {10.1103/RevModPhys.93.045003},
  url = {https://link.aps.org/doi/10.1103/RevModPhys.93.045003}
}

@article{Orus2014,
title = {A practical introduction to tensor networks: Matrix product states and projected entangled pair states},
journal = {Annals of Physics},
volume = {349},
pages = {117-158},
year = {2014},
issn = {0003-4916},
doi = {https://doi.org/10.1016/j.aop.2014.06.013},
url = {https://www.sciencedirect.com/science/article/pii/S0003491614001596},
author = {Román Orús}
}

@article{CrutchfieldYoung1989,
  author       = {J. P. Crutchfield and K. Young},
  title        = {Inferring statistical complexity},
  journal      = {Phys. Rev. Lett.},
  volume       = {63},
  pages        = {105--108},
  year         = {1989},
  doi          = {10.1103/PhysRevLett.63.105}
}

@article{ShaliziCrutchfield2001,
  author       = {C. R. Shalizi and J. P. Crutchfield},
  title        = {Computational mechanics: Pattern and prediction, structure and simplicity},
  journal      = {J. Stat. Phys.},
  volume       = {104},
  pages        = {817--879},
  year         = {2001},
  doi          = {10.1023/A:1010388907793}
}

@article{GuEtAl2012,
  author       = {M. Gu and K. Wiesner and E. Rieper and V. Vedral},
  title        = {Quantum mechanics can reduce the complexity of classical models},
  journal      = {Nat. Commun.},
  volume       = {3},
  pages        = {762},
  year         = {2012},
  doi          = {10.1038/ncomms1761}
}

@article{Zhang2022NonMarkovianSuperconducting,
  title = {Predicting Non-Markovian Superconducting-Qubit Dynamics from Tomographic Reconstruction},
  author = {Zhang, Haimeng and Pokharel, Bibek and Levenson-Falk, E.M. and Lidar, Daniel},
  journal = {Phys. Rev. Appl.},
  volume = {17},
  issue = {5},
  pages = {054018},
  numpages = {22},
  year = {2022},
  month = {May},
  publisher = {American Physical Society},
  doi = {10.1103/PhysRevApplied.17.054018},
  url = {https://link.aps.org/doi/10.1103/PhysRevApplied.17.054018}
}

@article{WhitePRL2023,
  title = {Filtering Crosstalk from Bath Non-Markovianity via Spacetime Classical Shadows},
  author = {White, G. A. L. and Modi, K. and Hill, C. D.},
  journal = {Phys. Rev. Lett.},
  volume = {130},
  issue = {16},
  pages = {160401},
  numpages = {6},
  year = {2023},
  month = {Apr},
  publisher = {American Physical Society},
  doi = {10.1103/PhysRevLett.130.160401},
  url = {https://link.aps.org/doi/10.1103/PhysRevLett.130.160401}
}

@inproceedings{GutoskiWatrous2007,
  author    = {Gus Gutoski and John Watrous},
  title     = {Toward a General Theory of Quantum Games},
  booktitle = {Proceedings of the 39th Annual ACM Symposium on Theory of Computing (STOC)},
  year      = {2007},
  pages     = {565--574},
  doi       = {10.1145/1250790.1250873},
  publisher = {ACM}
}

@article{helstrom1969,
  author  = {Carl W. Helstrom},
  title   = {Quantum Detection and Estimation Theory},
  journal = {Journal of Statistical Physics},
  volume  = {1},
  number  = {2},
  pages   = {231--252},
  year    = {1969}
}

@article{chiribella2008,
  author  = {Chiribella, G. and D'Ariano, G. M. and Perinotti, P.},
  title   = {Transforming quantum operations: quantum supermaps},
  journal = {EPL (Europhysics Letters)},
  volume  = {83},
  number  = {3},
  pages   = {30004},
  year    = {2008},
  doi     = {10.1209/0295-5075/83/30004},
  url     = {https://doi.org/10.1209/0295-5075/83/30004},
  eprint  = {0804.0180},
}

@book{holevo2011,
  author    = {Holevo, Alexander},
  title     = {Probabilistic and Statistical Aspects of Quantum Theory},
  series    = {Publications of the Scuola Normale Superiore},
  volume    = {1},
  publisher = {Edizioni della Normale Pisa},
  year      = {2011},
  doi       = {10.1007/978-88-7642-378-9},
  url       = {https://doi.org/10.1007/978-88-7642-378-9},
  isbn      = {978-88-7642-378-9}
}

@article{ohst2026,
  author = {Ohst, Ties-Albrecht and Zhang, Shijun and Nguyen, Hai Chau and Pl{\'a}vala, Martin and Quintino, Marco T{\'u}lio},
  title = {Characterising memory in quantum channel discrimination via constrained separability problems},
  journal = {Quantum},
  volume = {10},
  pages = {1988},
  year = {2026},
  doi = {10.22331/q-2026-01-28-1988},
  eprint = {arXiv:2411.08110}
}

@article{bae2015,
  author = {Bae, Joonwoo and Kwek, Leong-Chuan},
  title = {Quantum state discrimination and its applications},
  journal = {Journal of Physics A: Mathematical and Theoretical},
  volume = {48},
  number = {8},
  pages = {083001},
  year = {2015},
  doi = {10.1088/1751-8113/48/8/083001}
}

@article{taranto2021,
  author = {Taranto, Philip and Pollock, Felix A. and Milz, Simon and Tomamichel, Marco and Modi, Kavan},
  title = {Non-Markovian memory strength bounds quantum process recoverability},
  journal = {npj Quantum Information},
  volume = {7},
  pages = {149},
  year = {2021},
  doi = {10.1038/s41534-021-00481-4}
}

@article{milz2020,
  author = {Milz, Simon and Egloff, Daniel and Taranto, Philip and Theurer, Thomas and Plenio, Martin B. and Smirne, Andrea and Huelga, Susana F.},
  title = {When Is a Non-Markovian Quantum Process Classical?},
  journal = {Physical Review X},
  volume = {10},
  number = {4},
  pages = {041049},
  year = {2020},
  doi = {10.1103/PhysRevX.10.041049}
}

@article{giarmatzi2023,
  author = {Giarmatzi, Cristina and Costa, Fabio},
  title = {Witnessing quantum memory in stochastic processes},
  journal = {Quantum},
  volume = {7},
  pages = {1036},
  year = {2023}
}

@article{preskill2018nisq,
  author = {Preskill, John},
  title = {Quantum Computing in the NISQ era and beyond},
  journal = {Quantum},
  volume = {2},
  pages = {79},
  year = {2018}
}

@article{mironowicz2024sdp,
  author  = {Piotr Mironowicz},
  title   = {Semi-definite programming and quantum information},
  journal = {Journal of Physics A: Mathematical and Theoretical},
  volume  = {57},
  number  = {16},
  pages   = {163002},
  year    = {2024},
  doi     = {10.1088/1751-8121/ad2b85},
  url     = {https://doi.org/10.1088/1751-8121/ad2b85}
}

@article{Jiang2020AFI,
  title={A Faster Interior Point Method for Semidefinite Programming},
  author={Haotian Jiang and Tarun Kathuria and Yin Tat Lee and Swati Padmanabhan and Zhao Song},
  journal={2020 IEEE 61st Annual Symposium on Foundations of Computer Science (FOCS)},
  year={2020},
  pages={910-918},
  url={https://api.semanticscholar.org/CorpusID:221836388}
}

@misc{Wolf2012,
  author = {Michael M. Wolf},
  title = {Quantum Channels \& Operations: Guided Tour},
  year = {2012},
  note = {Lecture notes},
  url = {https://www-m5.ma.tum.de/foswiki/pub/M5/Allgemeines/MichaelWolf/QChannelLecture.pdf}
}

@article{Kossakowski2012,
  title = {Quantum-correlation breaking channels, broadcasting scenarios, and finite Markov chains},
  author = {Korbicz, J. K. and Horodecki, P. and Horodecki, R.},
  journal = {Phys. Rev. A},
  volume = {86},
  issue = {4},
  pages = {042319},
  numpages = {8},
  year = {2012},
  month = {Oct},
  publisher = {American Physical Society},
  doi = {10.1103/PhysRevA.86.042319},
  url = {https://link.aps.org/doi/10.1103/PhysRevA.86.042319}
}

@article{Temme2010,
  author    = {Temme, K. and Kastoryano, M. J. and Ruskai, Mary Beth and Wolf, Michael M. and Verstraete, Frank},
  title     = {The $-\alpha$-divergence and mixing times of quantum Markov processes},
  journal   = {Journal of Mathematical Physics},
  volume    = {51},
  number    = {12},
  pages     = {122201},
  year      = {2010},
  doi       = {10.1063/1.3511335}
}

@article{Hirche2022,
  author  = {Hirche, Christoph and Rouz{\'e}, Cambyse and Stilck Fran{\c{c}}a, Daniel},
  title   = {On contraction coefficients, partial orders and approximation of capacities for quantum channels},
  journal = {Quantum},
  volume  = {6},
  pages   = {862},
  year    = {2022},
  doi     = {10.22331/q-2022-11-28-862}
}

@article{Souissi2025exponential,
  author    = {Souissi, Abdessatar and Barhoumi, Abdessatar},
  title     = {An Exponential Mixing Condition for Quantum Channels: Application to Matrix Product States},
  journal   = {Quantum Information Processing},
  volume    = {24},
  pages     = {150},
  year      = {2025},
  doi       = {10.1007/s11128-025-04762-1}
}

@article{gutoski2012measure,
  title = {On a measure of distance for quantum strategies},
  author = {Gutoski, Gus},
  journal = {Journal of Mathematical Physics},
  volume = {53},
  number = {3},
  pages = {032202},
  year = {2012},
  doi = {10.1063/1.3693621},
  url = {https://doi.org/10.1063/1.3693621}
}

@book{watrous2018theory,
  title = {The Theory of Quantum Information},
  author = {Watrous, John},
  publisher = {Cambridge University Press},
  address = {Cambridge},
  year = {2018},
  doi = {10.1017/9781316848142},
  url = {https://doi.org/10.1017/9781316848142}
}

@article{chiribella2009theoretical,
  title = {Theoretical framework for quantum networks},
  author = {Chiribella, Giulio and D'Ariano, Giacomo Mauro and Perinotti, Paolo},
  journal = {Physical Review A},
  volume = {80},
  pages = {022339},
  year = {2009},
  doi = {10.1103/PhysRevA.80.022339},
  url = {https://doi.org/10.1103/PhysRevA.80.022339}
}

@article{pollock2018nonmarkovian,
  title = {Non-{M}arkovian quantum processes: Complete framework and efficient characterization},
  author = {Pollock, Felix A. and Rodr{\'i}guez-Rosario, C{\'e}sar and Frauenheim, Thomas and Paternostro, Mauro and Modi, Kavan},
  journal = {Physical Review A},
  volume = {97},
  pages = {012127},
  year = {2018},
  doi = {10.1103/PhysRevA.97.012127},
  url = {https://doi.org/10.1103/PhysRevA.97.012127}
}

@article{Milz2017,
  author  = {Milz, Simon and Pollock, Felix A. and Modi, Kavan},
  title   = {An Introduction to Operational Quantum Dynamics},
  journal = {Open Systems \& Information Dynamics},
  volume  = {24},
  number  = {04},
  pages   = {1740016},
  year    = {2017},
  doi     = {10.1142/S1230161217400169},
  url     = {https://doi.org/10.1142/S1230161217400169}
}

@article{milz2021quantum,
  title = {Quantum Stochastic Processes and Quantum non-{M}arkovian Phenomena},
  author = {Milz, Simon and Modi, Kavan},
  journal = {PRX Quantum},
  volume = {2},
  pages = {030201},
  year = {2021},
  doi = {10.1103/PRXQuantum.2.030201},
  url = {https://doi.org/10.1103/PRXQuantum.2.030201}
}

@article{taranto2024hierarchy,
  title = {Characterising the Hierarchy of Multi-time Quantum Processes with Classical Memory},
  author = {Taranto, Philip and Quintino, Marco T{\'u}lio and Murao, Mio and Milz, Simon},
  journal = {Quantum},
  volume = {8},
  pages = {1328},
  year = {2024},
  doi = {10.22331/q-2024-05-02-1328},
  url = {https://doi.org/10.22331/q-2024-05-02-1328}
}

@article{jorgensen2020discrete,
  title = {Discrete memory kernel for multitime correlations in non-{M}arkovian quantum processes},
  author = {J{\o}rgensen, Mathias R. and Pollock, Felix A.},
  journal = {Physical Review A},
  volume = {102},
  pages = {052206},
  year = {2020},
  doi = {10.1103/PhysRevA.102.052206},
  url = {https://doi.org/10.1103/PhysRevA.102.052206}
}

@article{dowling2024capturing,
  title = {Capturing Long-Range Memory Structures with Tree-Geometry Process Tensors},
  author = {Dowling, Neil and Modi, Kavan and Mu{\~n}oz, Roberto N. and Singh, Sukhbinder and White, Gregory A. L.},
  journal = {Physical Review X},
  volume = {14},
  pages = {041018},
  year = {2024},
  doi = {10.1103/PhysRevX.14.041018},
  url = {https://doi.org/10.1103/PhysRevX.14.041018}
}

@article{ciccarello2022collision,
  title = {Quantum collision models: Open system dynamics from repeated interactions},
  author = {Ciccarello, Francesco and Lorenzo, Salvatore and Giovannetti, Vittorio and Palma, G. Massimo},
  journal = {Physics Reports},
  volume = {954},
  pages = {1--70},
  year = {2022},
  doi = {10.1016/j.physrep.2022.01.001},
  url = {https://doi.org/10.1016/j.physrep.2022.01.001}
}

@article{cygorek2022compression,
  title = {Simulation of open quantum systems by automated compression of arbitrary environments},
  author = {Cygorek, Moritz and Cosacchi, Matteo and Vagov, Alexander and Axt, V. M. and Lovett, Brendon W. and Keeling, Jonathan and Gauger, Erik M.},
  journal = {Nature Physics},
  volume = {18},
  pages = {662--668},
  year = {2022},
  doi = {10.1038/s41567-022-01544-9},
  url = {https://doi.org/10.1038/s41567-022-01544-9}
}

@article{fux2021efficient,
  title = {Efficient Exploration of Hamiltonian Parameter Space for Optimal Control of Non-{M}arkovian Open Quantum Systems},
  author = {Fux, Gerald E. and Butler, Eoin P. and Eastham, Paul R. and Lovett, Brendon W. and Keeling, Jonathan},
  journal = {Physical Review Letters},
  volume = {126},
  pages = {200401},
  year = {2021},
  doi = {10.1103/PhysRevLett.126.200401},
  url = {https://doi.org/10.1103/PhysRevLett.126.200401}
}

@article{ivander2024unified,
  author  = {Ivander, Felix and Lindoy, Lachlan P. and Lee, Joonho},
  title   = {Unified framework for open quantum dynamics with memory},
  journal = {Nature Communications},
  volume  = {15},
  pages   = {8087},
  year    = {2024},
  doi     = {10.1038/s41467-024-52081-3},
  url     = {https://doi.org/10.1038/s41467-024-52081-3}
}

@incollection{fujii2021quantum,
  title = {Quantum reservoir computing: A reservoir approach toward quantum machine learning on near-term quantum devices},
  author = {Fujii, Keisuke and Nakajima, Kohei},
  booktitle = {Reservoir Computing: Theory, Physical Implementations, and Applications},
  editor = {Nakajima, Kohei and Fischer, Ingo},
  publisher = {Springer Singapore},
  address = {Singapore},
  pages = {423--450},
  year = {2021},
  doi = {10.1007/978-981-13-1687-6_18},
  url = {https://doi.org/10.1007/978-981-13-1687-6_18}
}

@article{nakajima2019boosting,
  title = {Boosting Computational Power through Spatial Multiplexing in Quantum Reservoir Computing},
  author = {Nakajima, Kohei and Fujii, Keisuke and Negoro, Makoto and Mitarai, Kosuke and Kitagawa, Masahiro},
  journal = {Physical Review Applied},
  volume = {11},
  pages = {034021},
  year = {2019},
  doi = {10.1103/PhysRevApplied.11.034021},
  url = {https://doi.org/10.1103/PhysRevApplied.11.034021}
}

@article{mujal2021opportunities,
  title = {Opportunities in Quantum Reservoir Computing and Extreme Learning Machines},
  author = {Mujal, Pere and Mart{\'i}nez-Pe{\~n}a, Rodrigo and Nokkala, Johannes and Garc{\'i}a-Beni, Jorge and Giorgi, Gian Luca and Soriano, Miguel C. and Zambrini, Roberta},
  journal = {Advanced Quantum Technologies},
  volume = {4},
  pages = {2100027},
  year = {2021},
  doi = {10.1002/qute.202100027},
  url = {https://doi.org/10.1002/qute.202100027}
}

@article{govia2021reservoir,
  title = {Quantum reservoir computing with a single nonlinear oscillator},
  author = {Govia, L. C. G. and Ribeill, G. J. and Rowlands, G. E. and Krovi, H. K. and Ohki, T. A.},
  journal = {Physical Review Research},
  volume = {3},
  pages = {013077},
  year = {2021},
  doi = {10.1103/PhysRevResearch.3.013077},
  url = {https://doi.org/10.1103/PhysRevResearch.3.013077}
}

@article{suzuki2022natural,
  title = {Natural quantum reservoir computing for temporal information processing},
  author = {Suzuki, Yudai and Gao, Qi and Pradel, Ken C. and Yasuoka, Kenji and Yamamoto, Naoki},
  journal = {Scientific Reports},
  volume = {12},
  pages = {1353},
  year = {2022},
  doi = {10.1038/s41598-022-05061-w},
  url = {https://doi.org/10.1038/s41598-022-05061-w}
}

@article{wilde2020amortized,
  title = {Amortized channel divergence for asymptotic quantum channel discrimination},
  author = {Wilde, Mark M. and Berta, Mario and Hirche, Christoph and Kaur, Eneet},
  journal = {Letters in Mathematical Physics},
  volume = {110},
  pages = {2277--2336},
  year = {2020},
  doi = {10.1007/s11005-020-01297-7},
  url = {https://doi.org/10.1007/s11005-020-01297-7}
}

@article{wang2019asymmetric,
  title = {Resource theory of asymmetric distinguishability for quantum channels},
  author = {Wang, Xin and Wilde, Mark M.},
  journal = {Physical Review Research},
  volume = {1},
  pages = {033169},
  year = {2019},
  doi = {10.1103/PhysRevResearch.1.033169},
  url = {https://doi.org/10.1103/PhysRevResearch.1.033169}
}

@article{katariya2021geometric,
  title = {Geometric distinguishability measures limit quantum channel estimation and discrimination},
  author = {Katariya, Vishal and Wilde, Mark M.},
  journal = {Quantum Information Processing},
  volume = {20},
  pages = {78},
  year = {2021},
  doi = {10.1007/s11128-021-02992-7},
  url = {https://doi.org/10.1007/s11128-021-02992-7}
}

@article{nakahira2021general,
  title = {Generalized quantum process discrimination problems},
  author = {Nakahira, Kenji and Kato, Kentaro},
  journal = {Physical Review A},
  volume = {103},
  pages = {062606},
  year = {2021},
  doi = {10.1103/PhysRevA.103.062606},
  url = {https://doi.org/10.1103/PhysRevA.103.062606}
}

@article{hirche2023network,
  title = {Quantum Network Discrimination},
  author = {Hirche, Christoph},
  journal = {Quantum},
  volume = {7},
  pages = {1064},
  year = {2023},
  doi = {10.22331/q-2023-07-25-1064},
  url = {https://doi.org/10.22331/q-2023-07-25-1064}
}

@article{portmann2022security,
  title = {Security in quantum cryptography},
  author = {Portmann, Christopher and Renner, Renato},
  journal = {Reviews of Modern Physics},
  volume = {94},
  pages = {025008},
  year = {2022},
  doi = {10.1103/RevModPhys.94.025008},
  url = {https://doi.org/10.1103/RevModPhys.94.025008}
}

@article{chiribella2008memory,
  title = {Memory Effects in Quantum Channel Discrimination},
  author = {Chiribella, Giulio and D'Ariano, Giacomo Mauro and Perinotti, Paolo},
  journal = {Physical Review Letters},
  volume = {101},
  pages = {180501},
  year = {2008},
  doi = {10.1103/PhysRevLett.101.180501},
  url = {https://doi.org/10.1103/PhysRevLett.101.180501}
}

@article{harrow2010adaptive,
  title = {Adaptive versus nonadaptive strategies for quantum channel discrimination},
  author = {Harrow, Aram W. and Hassidim, Avinatan and Leung, Debbie W. and Watrous, John},
  journal = {Physical Review A},
  volume = {81},
  pages = {032339},
  year = {2010},
  doi = {10.1103/PhysRevA.81.032339},
  url = {https://doi.org/10.1103/PhysRevA.81.032339}
}

@article{jencova2016conditions,
  title = {Conditions for optimal input states for discrimination of quantum channels},
  author = {Jen{\v{c}}ov{\'a}, Anna and Pl{\'a}vala, Martin},
  journal = {Journal of Mathematical Physics},
  volume = {57},
  pages = {122203},
  year = {2016},
  doi = {10.1063/1.4972286},
  url = {https://doi.org/10.1063/1.4972286}
}

@article{bavaresco2021strict,
  title = {Strict Hierarchy between Parallel, Sequential, and Indefinite-Causal-Order Strategies for Channel Discrimination},
  author = {Bavaresco, Jessica and Murao, Mio and Quintino, Marco T{\'u}lio},
  journal = {Physical Review Letters},
  volume = {127},
  pages = {200504},
  year = {2021},
  doi = {10.1103/PhysRevLett.127.200504},
  url = {https://doi.org/10.1103/PhysRevLett.127.200504}
}

@article{salek2022usefulness,
  title = {Usefulness of adaptive strategies in asymptotic quantum channel discrimination},
  author = {Salek, Farzin and Hayashi, Masahito and Winter, Andreas},
  journal = {Physical Review A},
  volume = {105},
  pages = {022419},
  year = {2022},
  doi = {10.1103/PhysRevA.105.022419},
  url = {https://doi.org/10.1103/PhysRevA.105.022419}
}

@article{nakahira2021restricted,
  title = {Quantum process discrimination with restricted strategies},
  author = {Nakahira, Kenji},
  journal = {Physical Review A},
  volume = {104},
  pages = {062609},
  year = {2021},
  doi = {10.1103/PhysRevA.104.062609},
  url = {https://doi.org/10.1103/PhysRevA.104.062609}
}

@misc{lyu2026variational,
  title = {Variational Quantum Dimension Reduction for Recurrent Quantum Models},
  author = {Lyu, Chufan and Wang, Ximing and Gu, Mile and Elliott, Thomas J. and Yang, Chengran},
  year = {2026},
  eprint = {2603.09567},
  archivePrefix = {arXiv},
  primaryClass = {quant-ph},
  url = {https://arxiv.org/abs/2603.09567}
}

@article{white2020demonstration,
  title = {Demonstration of non-{M}arkovian process characterisation and control on a quantum processor},
  author = {White, G. A. L. and Hill, C. D. and Pollock, F. A. and Hollenberg, L. C. L. and Modi, K.},
  journal = {Nature Communications},
  volume = {11},
  pages = {6301},
  year = {2020},
  doi = {10.1038/s41467-020-20113-3},
  url = {https://doi.org/10.1038/s41467-020-20113-3}
}

@article{zhang2022predicting,
  title = {Predicting non-{M}arkovian superconducting-qubit dynamics from tomographic reconstruction},
  author = {Zhang, Haimeng and Pokharel, Bibek and Levenson-Falk, E. M. and Lidar, Daniel A.},
  journal = {Physical Review Applied},
  volume = {17},
  pages = {054018},
  year = {2022},
  doi = {10.1103/PhysRevApplied.17.054018},
  url = {https://doi.org/10.1103/PhysRevApplied.17.054018}
}

@article{nakamura2024gate,
  title = {Gate Operations for Superconducting Qubits and Non-{M}arkovianity},
  author = {Nakamura, Kiyoto and Ankerhold, Joachim},
  journal = {Physical Review Research},
  volume = {6},
  pages = {033215},
  year = {2024},
  doi = {10.1103/PhysRevResearch.6.033215},
  url = {https://doi.org/10.1103/PhysRevResearch.6.033215}
}

@article{OrtegaTaberner2024Unifying,
  author  = {Ortega-Taberner, Carlos and O'Neill, Eoin and Butler, Eoin and Fux, Gerald E. and Eastham, P. R.},
  title   = {Unifying methods for optimal control in non-Markovian quantum systems via process tensors},
  journal = {The Journal of Chemical Physics},
  volume  = {161},
  number  = {12},
  pages   = {124119},
  year    = {2024},
  doi     = {10.1063/5.0226031},
  url     = {https://doi.org/10.1063/5.0226031}
}

@article{AttalPautrat2006,
  author  = {Attal, St{\'e}phane and Pautrat, Yan},
  title   = {From Repeated to Continuous Quantum Interactions},
  journal = {Annales Henri Poincar{\'e}},
  volume  = {7},
  number  = {1},
  pages   = {59--104},
  year    = {2006},
  doi     = {10.1007/s00023-005-0242-8},
  url     = {https://doi.org/10.1007/s00023-005-0242-8}
}

@article{BruneauJoyeMerkli2014,
  author        = {Bruneau, Laurent and Joye, Alain and Merkli, Marco},
  title         = {Repeated Interactions in Open Quantum Systems},
  journal       = {Journal of Mathematical Physics},
  volume        = {55},
  number        = {7},
  pages         = {075204},
  year          = {2014},
  doi           = {10.1063/1.4879240},
  url           = {https://doi.org/10.1063/1.4879240},
  eprint        = {1305.2472},
  archivePrefix = {arXiv},
  primaryClass  = {math-ph}
}

@article{KretschmannWerner2005,
  author        = {Kretschmann, Dennis and Werner, Reinhard F.},
  title         = {Quantum Channels with Memory},
  journal       = {Physical Review A},
  volume        = {72},
  number        = {6},
  pages         = {062323},
  year          = {2005},
  doi           = {10.1103/PhysRevA.72.062323},
  url           = {https://doi.org/10.1103/PhysRevA.72.062323},
  eprint        = {quant-ph/0502106},
  archivePrefix = {arXiv},
  primaryClass  = {quant-ph}
}

@book{BreuerPetruccione2002,
  author    = {Breuer, Heinz-Peter and Petruccione, Francesco},
  title     = {The Theory of Open Quantum Systems},
  publisher = {Oxford University Press},
  address   = {Oxford},
  year      = {2002},
  isbn      = {978-0-19-852063-4}
}

@book{Nielsen_Chuang_2010, 
place={Cambridge}, 
title={Quantum Computation and Quantum Information: 10th Anniversary Edition}, 
publisher={Cambridge University Press}, 
author={Nielsen, Michael A. and Chuang, Isaac L.}, 
year={2010}}

@article{Kitaev1997QuantumComputations,
  author  = {Kitaev, A. Yu.},
  title   = {Quantum computations: algorithms and error correction},
  journal = {Russian Mathematical Surveys},
  volume  = {52},
  number  = {6},
  pages   = {1191--1249},
  year    = {1997},
  doi     = {10.1070/RM1997v052n06ABEH002155}
}

@article{terhal2015quantum,
  author = {Terhal, Barbara M.},
  title = {Quantum Error Correction for Quantum Memories},
  journal = {Reviews of Modern Physics},
  volume = {87},
  pages = {307--346},
  year = {2015},
  doi = {10.1103/RevModPhys.87.307}
}

@article{heshami2016quantum,
  author = {Heshami, Khabat and England, Duncan G. and Humphreys, Peter C. and Bustard, Philip J. and Acosta, Victor M. and Nunn, Joshua and Sussman, Benjamin J.},
  title = {Quantum Memories: Emerging Applications and Recent Advances},
  journal = {Journal of Modern Optics},
  volume = {63},
  number = {20},
  pages = {2005--2028},
  year = {2016},
  doi = {10.1080/09500340.2016.1148212}
}

@misc{keeling2025processTensorApproaches,
  title = {Process Tensor Approaches to Non-Markovian Quantum Dynamics},
  author = {Keeling, Jonathan and Stoudenmire, E. Miles and Ba{\~n}uls, Mari-Carmen and Reichman, David R.},
  year = {2025},
  eprint = {2509.07661},
  archivePrefix = {arXiv},
  primaryClass = {quant-ph},
  doi = {10.48550/arXiv.2509.07661},
  url = {https://arxiv.org/abs/2509.07661},
  note = {Invited perspective}
}

@article{fux2024oqupy,
  title = {{OQuPy}: A {Python} package to efficiently simulate non-{Markovian} open quantum systems with process tensors},
  author = {Fux, Gerald E. and Fowler-Wright, Piper and Beckles, Joel and Butler, Eoin P. and Eastham, Paul R. and Gribben, Dominic and Keeling, Jonathan and Kilda, Dainius and Kirton, Peter and Lawrence, Ewen D. C. and Lovett, Brendon W. and {O'Neill}, Eoin and Strathearn, Aidan and de Wit, Roosmarijn},
  journal = {Journal of Chemical Physics},
  volume = {161},
  number = {12},
  pages = {124108},
  year = {2024},
  doi = {10.1063/5.0225367},
  url = {https://doi.org/10.1063/5.0225367}
}

@article{cygorek2024ace,
  title = {{ACE}: A general-purpose non-{Markovian} open quantum systems simulation toolkit based on process tensors},
  author = {Cygorek, Moritz and Gauger, Erik M.},
  journal = {Journal of Chemical Physics},
  volume = {161},
  number = {7},
  pages = {074111},
  year = {2024},
  doi = {10.1063/5.0221182},
  url = {https://doi.org/10.1063/5.0221182}
}

@article{sannia2024dissipation,
  title = {Dissipation as a resource for {Quantum} {Reservoir} {Computing}},
  author = {Sannia, Antonio and Mart{\'i}nez-Pe{\~n}a, Rodrigo and Soriano, Miguel C. and Giorgi, Gian Luca and Zambrini, Roberta},
  journal = {Quantum},
  volume = {8},
  pages = {1291},
  year = {2024},
  month = mar,
  doi = {10.22331/q-2024-03-20-1291},
  url = {https://doi.org/10.22331/q-2024-03-20-1291}
}

@article{zhang2025learningForecasting,
  title = {Learning and forecasting open quantum dynamics with correlated noise},
  author = {Zhang, Xinfang and Wu, Zhihao and White, Gregory A. L. and Xiang, Zhongcheng and Hu, Shun and Peng, Zhihui and Liu, Yong and Zheng, Dongning and Fu, Xiang and Huang, Anqi and Poletti, Dario and Modi, Kavan and Wu, Junjie and Deng, Mingtang and Guo, Chu},
  journal = {Communications Physics},
  volume = {8},
  pages = {29},
  year = {2025},
  doi = {10.1038/s42005-025-01944-2},
  url = {https://doi.org/10.1038/s42005-025-01944-2}
}

@article{martinezpena2025inputDependence,
  title = {Input-dependence in quantum reservoir computing},
  author = {Mart{\'i}nez-Pe{\~n}a, Rodrigo and Ortega, Juan-Pablo},
  journal = {Physical Review E},
  volume = {111},
  number = {6},
  pages = {065306},
  year = {2025},
  doi = {10.1103/3775-4hfd},
  url = {https://doi.org/10.1103/3775-4hfd},
  eprint = {2412.08322},
  archivePrefix = {arXiv},
  primaryClass = {quant-ph}
}

@article{wringe2025reservoirBenchmarks,
  title = {Reservoir computing benchmarks: a tutorial review and critique},
  author = {Wringe, Chester and Stepney, Susan and Trefzer, Martin},
  journal = {International Journal of Parallel, Emergent and Distributed Systems},
  volume = {40},
  number = {4},
  pages = {313--351},
  year = {2025},
  doi = {10.1080/17445760.2025.2472211},
  url = {https://doi.org/10.1080/17445760.2025.2472211}
}

@article{vieira_temporal_2022,
	title = {Temporal correlations in the simplest measurement sequences},
	volume = {6},
	issn = {2521-327X},
	url = {https://quantum-journal.org/papers/q-2022-01-18-623/},
	doi = {10.22331/q-2022-01-18-623},
	language = {en},
	urldate = {2024-08-20},
	journal = {Quantum},
	author = {Vieira, Lucas B. and Budroni, Costantino},
	month = jan,
	year = {2022},
	pages = {623},
}

@article{Zonnios2025QuantumGeneration,
  author  = {Zonnios, Magdalini and Boyd, Alec and Binder, Felix C.},
  title   = {Quantum generation of stochastic processes: spectral invariants and memory bounds},
  journal = {New Journal of Physics},
  volume  = {27},
  number  = {6},
  pages   = {064507},
  year    = {2025},
  doi     = {10.1088/1367-2630/addc10},
  url     = {https://doi.org/10.1088/1367-2630/addc10}
}

@article{white2025unifyingPRX,
  title = {Unifying Non-Markovian Characterization with an Efficient and Self-Consistent Framework},
  author = {White, G. A. L. and Jurcevic, P. and Hill, C. D. and Modi, K.},
  journal = {Phys. Rev. X},
  volume = {15},
  issue = {2},
  pages = {021047},
  numpages = {43},
  year = {2025},
  month = {May},
  publisher = {American Physical Society},
  doi = {10.1103/PhysRevX.15.021047},
  url = {https://link.aps.org/doi/10.1103/PhysRevX.15.021047}
}

@article{whiteNonMarkovian2022PRX,
  title = {Non-Markovian Quantum Process Tomography},
  author = {White, G.A.L. and Pollock, F.A. and Hollenberg, L.C.L. and Modi, K. and Hill, C.D.},
  journal = {PRX Quantum},
  volume = {3},
  issue = {2},
  pages = {020344},
  numpages = {30},
  year = {2022},
  month = {May},
  publisher = {American Physical Society},
  doi = {10.1103/PRXQuantum.3.020344},
  url = {https://link.aps.org/doi/10.1103/PRXQuantum.3.020344}
}

@article{Thompson2017UsingQuantumTheory,
  author  = {Thompson, Jayne and Garner, Andrew J. P. and Vedral, Vlatko and Gu, Mile},
  title   = {Using quantum theory to simplify input--output processes},
  journal = {npj Quantum Information},
  year    = {2017},
  volume   = {3},
  number   = {1},
  pages    = {6},
  doi      = {10.1038/s41534-016-0001-3},
  url      = {https://doi.org/10.1038/s41534-016-0001-3}
}

@article{YangEtAl2025QuantumDimensionReduction,
  author  = {Yang, Chengran and Florido-Llin{\`a}s, Marta and Gu, Mile and Elliott, Thomas J.},
  title   = {Dimension reduction in quantum sampling of stochastic processes},
  journal = {npj Quantum Information},
  volume   = {11},
  number   = {1},
  pages    = {34},
  year     = {2025},
  doi      = {10.1038/s41534-025-00978-2}
}

@article{binderPRLpractical2018,
  title = {Practical Unitary Simulator for Non-Markovian Complex Processes},
  author = {Binder, Felix C. and Thompson, Jayne and Gu, Mile},
  journal = {Phys. Rev. Lett.},
  volume = {120},
  issue = {24},
  pages = {240502},
  numpages = {6},
  year = {2018},
  month = {Jun},
  publisher = {American Physical Society},
  doi = {10.1103/PhysRevLett.120.240502},
  url = {https://link.aps.org/doi/10.1103/PhysRevLett.120.240502}
}

@misc{kechrimparis2025quantumagentscomplexity,
title = {How Quantum Agents Can Change Which Strategies Are More Complex},
author = {Kechrimparis, Spiros and Barnett, Nix and Gu, Mile and Kwon, Hyukjoon},
year = {2025},
eprint = {2508.08092},
archivePrefix = {arXiv},
primaryClass = {quant-ph},
doi = {10.48550/arXiv.2508.08092}
}

@misc{lumbreras2026reinforcement,
title = {Reinforcement learning for quantum processes with memory},
author = {Lumbreras, Josep and Huang, Ruo Cheng and Hu, Yanglin and Fanizza, Marco and Gu, Mile},
year = {2026},
eprint = {2603.25138},
archivePrefix = {arXiv},
primaryClass = {quant-ph},
doi = {10.48550/arXiv.2603.25138}
}

\newpage

\appendix

\section{Choi conventions, comb constraints, and tester constraints}
\label{app:choi-combs-testers}

We follow the process-tensor / quantum-comb formalism
\cite{chiribella2009theoretical,pollock2018nonmarkovian,MilzPollockModi2016}. The salient points of this framework and this notation are specified in this Appendix.

\subsection{Quantum combs and their Choi representation}
\label{app:combs}

An $N$-step quantum process (comb) $\mathbf{P}^{N}$ can be realized by a sequence of completely positive and trace preserving (CPTP) maps
$\{\mathcal{P}^{(k:k+1)}_{SE}\}_{k=1}^{N}$ acting on a system and an inaccessible environment, connected only through environmental degrees of freedom, with the final environment space traced out. Each map $\mathcal{P}^{(k:k+1)}_{SE}$ acts on the system and environment spaces corresponding to the input output space at that time step,
\begin{align}
    \mathcal{P}^{(k:k+1)}_{SE}:S_kE_{k}\mapsto S_k'E_{k+1}
\end{align}
where $S_k$ and $S_k'$ denotes the system input and output spaces at time $k$, and $E_{k}$ and $E_{k+1}$ are the input and output on the environment as per Fig.~\ref{fig:tester_and_comb}. Then, writing $\circ_{E_k}$ for composition over the $k^{\text{th}}$ environment only, an $N$-step comb may be expressed as

\begin{widetext}
    \begin{align}
\mathbf{P}^{N}
:=
\tr_{E_{N+1}}\!\!\left(\!
\mathcal{P}^{(N-1:N)}_{SE}\circ_{E_{N-1}}\!\mathcal{P}^{(N-2:N-1)}_{SE}\circ_{E_2}\dots\circ_{E_1}\!\mathcal{P}^{(0:1)}_{SE}\,[\rho^{\mathbf{P}}_{SE,0}]
\!\right)
\end{align}
\end{widetext}

where the trace is taken over the final environment space $E_N$, $\circ{E_k}$ is the composition over the $k^{\text{th}}$ environment. The state $\rho^{\mathbf{P}}_{SE,0}$ is the initial system-environment state specified by the process. For notational simplicity we write $\mathcal{P}^{(k:k+1)}_{SE}\equiv \mathcal{P}^{(k)}$. The comb acts on the open system spaces
\begin{align}
S'_0, S_1, S'_1, S_2,  \dots , S'_{N-1},  S_N.
\end{align}
We use the Choi--Jamio{\l}kowski representation. For a single-step map $\mathcal{P}^{(k)}$ the Choi operator is
\begin{align}
\Upsilon^{\mathcal{P}^{(k)}}
:=
\sum_{i,j}
(\mathcal{P}^{(k)} \otimes \mathcal I)
\bigl(
|i\rangle\!\langle j| \otimes |i\rangle\!\langle j|
\bigr).
\end{align}
Multi-time Choi operators are obtained by linking consecutive Choi operators along the shared environment degrees of freedom, via the link product $\star$~\cite{chiribella2009theoretical}. For two consecutive steps,
\begin{align}\label{eqn:env_link_app}
    \begin{split}
        \Upsilon^{\mathcal{P}^{(k)}}\star \Upsilon^{\mathcal{P}^{(k-1)}}
        :=
        &\tr_{E_{k}}\!\Big[
        (\Upsilon^{\mathcal{P}^{(k-1)}}\otimes\mathds{1}_{E_{k+1}S_{k}S'_{k}})
        \\&(\mathds{1}_{E_{k-1}S_{k-1}S'_{k-1}}\otimes(\Upsilon^{\mathcal{P}^{(k)}})^{T_{E_k}})
        \Big],
    \end{split}
\end{align}
where the trace $\tr_{E_k}$ and transpose $T_{E_k}$ are taken over the joint environment space across which the maps are linked (the input space of the map $\mathcal{P}^{(k-1)}$ with the output space of $\mathcal{P}^{(k)}$). Iterating this construction and linking with an initial state on $E_1$ and tracing out the final environment $E_N$ yields the $N$-slot Choi operator (see also \cite{Berk2021ResourceTheoriesOf}),
\begin{align}
\Upsilon^{\mathbf{P}^{N}}
:=
\tr_{E_{N}}\!\Big(
\Upsilon^{\mathcal{P}^{(N)}} \star \Upsilon^{\mathcal{P}^{(N-1)}}\star\dots\star
\Upsilon^{\mathcal{P}^{(1)}}\star\rho_{SE,0}^{\mathbf{P}}
\Big).
\end{align}

Since the process is causal and trace preserving, its Choi operator satisfies the multi-time constraints
\begin{subequations}\label{eq:causality_constraints_app}
    \begin{align}
    \tr_{S_k}\!\left(\Upsilon^{\mathbf{P}^{k}}\right)
    &= \Upsilon^{\mathbf{P}^{0:{k-1}}}\otimes\mathds{1}_{S_k'},
    \qquad 1\leq k\leq N,
    \\
    \Upsilon^{\mathbf{P}^{k}}&\succeq 0,
    \qquad 1\leq k\leq N.
\end{align}
\end{subequations}
These generalize complete positivity and trace preservation to the multi-time setting.

\subsection{Testers as multi-time instruments}
\label{app:testers}

A tester is the multi-time analogue of a quantum instrument. It specifies a correlated sequence of operations (possibly with internal memory) wired into the slots of the process and returns a classical outcome. Formally, a tester is a collection of completely positive (CP), trace-nonincreasing combs $\{\mathbf{T}^{N}_j\}_j$ satisfying
\begin{align}
\sum_j \mathbf{T}^{N}_j = \mathbf{T}^{N},
\end{align}
where $\mathbf{T}^{N}$ is a deterministic tester (a causal and trace preserving comb).

With the same labeling of system spaces as above, the constraints for a deterministic tester can be written as
\begin{subequations}\label{eqn:tester_constraints_app}
\begin{align}
\tr_{{S_{k+1}}}\!\left(\Upsilon^{\mathbf{T}^{k}}\right)
&=
\Upsilon^{\mathbf{T}^{0:{k-1}}}\otimes\mathds{1}_{{S'_k}},
\qquad 0\leq k\leq N,
\\
\Upsilon^{\mathbf{T}^{k}_j} &\succeq 0,
\qquad \forall\,k,j.
\end{align}
\end{subequations}
Operationally, given a process comb $\mathbf{P}^{N}$, contracting it with a tester element produces the (subnormalized) output state associated with outcome $j$; the outcome probability is the trace of that state. In the Choi representation this reduces to the Hilbert--Schmidt pairing,
\begin{align}
\Pr(j|\mathbf{P}^{N}, \mathbf{T}^{N})
=
\tr\!\left(\Upsilon^{\mathbf{P}^{N}}\,(\Upsilon^{\mathbf{T}^{N}_j})^T\right).
\end{align}

\begin{lemma}\label{lem:Realisation of testers}
\cite{GutoskiWatrous2007,chiribella2009theoretical} Every $N$-step tester $\{\mathbf{T}^{N}_j\}_j$ admits an operational realisation as a sequential interaction between the probed system and an internal memory register. Equivalently, it can be realized either as a sequence of CPTP maps $\{\mathcal{T}^{(k:k+1)}_{SM}\}_{k=0}^{N}$ linked through memory $m$, followed by a final quantum instrument $\{\mathcal{J}^{(N+1:N+2)}_{SM,j}\}_j$ where each map in the set is a completely positive trace non-increasing map corresponding to outcome $j$, or as a sequence of quantum instruments whose overall outcome statistics reproduce the tester.
\end{lemma}

From Lemma~\eqref{lem:Realisation of testers}, a tester element can thus be written as
\begin{widetext}
    \begin{align}
\mathbf{T}^{N}_j
:=
\tr_{M_{N+2}}\!\left(
\mathcal{J}_{j}^{(N+1:N+2)}\!\circ_{M_{N+1}}\! \mathcal{T}_{SM}^{(N:N+1)}\!\circ_{M_{N}}\!\dots\circ_{M_{0}}\! \mathcal{T}_{SM}^{(0:1)}[\rho^{\mathsf{T}}_{M,0}]
\right),
\end{align}
\end{widetext}
where $\circ_{M_{k}}$ denotes composition over the memory space at the $k^{\text{th}}$ time step (see Fig.~\ref{fig:tester_and_comb}). The state $\rho^{\mathsf{T}}_{M,0}$ is the initial memory state specified by the process. Since the channel $\mathcal{T}^{0:1}_{SM}$ is a map on the system and memory, its action on the memory state can also be seen as a partial composition over the memory $\mathcal{T}^{0:1}_{SM}[\rho^{\mathsf{T}}_{M,0}]\equiv \mathcal{T}^{0:1}_{SM}\circ_M\rho^{\mathsf{T}}_{M,0}$ which leaves the system degrees of freedom as free indices that the initial state of the system (possibly entangled with an environment) can contract with.
Each map acts on the system and memory 
\begin{align}
    \mathcal{T}^{(k:k+1)}_{SM} :
 ~&S_k'M_k\mapsto S_{k+2}M_{k+2},
\end{align}

This mirrors the internal structure of the process to which it is applied. In what follows we again adopt the notation $\mathcal{T}_{SM}^{(k:k+1)}\equiv \mathcal{T}_{}^{(k)}$.  The completely positive map $\mathcal{J}_{j}^{(N+1:N+2)}$ corresponds to observed outcome $j$, for example the length $N+1$ sequence $x_{0:N}=(x_0,\dots,x_N)$ observed by monitoring the process at each time. This construction allows one to defer all intermediate measurements so that any sequence of instruments applied at different time steps can be equivalently represented by a single measurement at the end, whose outcomes $j$ label the entire sequence of intermediate outcomes. Consequently, the number of possible outcomes grows as the product of the number of outcomes at each time step, and is therefore generally exponential in $N$. The Choi operator of a tester element may be written as a link-product chain,
\begin{align}
\Upsilon^{\mathbf{T}^{N}_j}
=
\Upsilon^{\mathcal{J}^{(N+1)}_j}\star
\Upsilon^{\mathcal{T}^{(N)}}\star
\Upsilon^{\mathcal{T}^{(N-1)}}\star\dots\star
\Upsilon^{\mathcal{T}^{(0)}}\star
\rho_{M,0}^{\mathsf{T}}.
\end{align}
The link product between consecutive tester-step Choi operators is defined analogously to Eq.~\eqref{eqn:env_link_app}, but now over the shared memory space:
\begin{align}\label{eqn:anc_link_app}
\begin{split}
\Upsilon^{\mathcal{T}^{(k)}}\star \Upsilon^{\mathcal{T}^{(k-1)}}
:=
\tr_{M_{k}}\!&\Big[
(\Upsilon^{\mathcal{T}^{(k-1)}}\otimes\mathds{1}_{S'_{k}S_{k+1}M_{k+1}})\\&
(\mathds{1}_{M_{k-1}S'_{k-1}S_k}\otimes(\Upsilon^{\mathcal{T}^{(k+1)}})^{T_{M_k}})
\Big].
\end{split}
\end{align}

For all valid processes $\mathbf{P}^{N}$, tester elements must satisfy
\begin{subequations}\label{eq:tester_element_constraints_app}
    \begin{align}
0\leq\tr(\Upsilon^{\mathbf{T}^{N}_j}\Upsilon^{\mathbf{P}^{N}})&\leq 1,
\\
\sum_j \tr(\Upsilon^{\mathbf{T}^{N}_j}\Upsilon^{\mathbf{P}^{N}}) &= 1,
\end{align}
\end{subequations}
i.e.\ outcome probabilities lie in $[0,1]$ and sum to $1$.

\subsection{Binary testers for process distinguishability}\label{app:binary_testers}
For discriminating between two processes $\mathbf{P}^N$ and $\mathbf{Q}^N$, any strategy using tester $\{\mathbf{T}_j^N\}_j$ with outcomes $j$ is specified by a decision rule assigning each $j$ (typically exponential in $N$) to a guess of the underlying process. The optimal decision rule assigns each outcome $j$ to the hypothesis with larger likelihood, i.e., compares $\Pr(j|\mathbf{P}^N,\mathbf{T}^N)$ and $\Pr(j|\mathbf{Q}^N,\mathbf{T}^N)$. This induces a partition of outcomes into two sets, corresponding to the two hypotheses.

Defining a new tester with two outcomes by coarse-graining,
\begin{align}
\mathbf{T}_0^N := \sum_{j \in S} \mathbf{T}_j^N,
\qquad
\mathbf{T}_1^N := \sum_{j \notin S} \mathbf{T}_j^N   
\end{align}
which gives the binary tester $\mathbf{T}_{0|1}^N$.

\subsection{Strategy-norm distance reduces to trace and diamond-norm distances}\label{app:strategy-norm-reduces-trace-diamond}
The quantity $d_{\mathrm{str}}^{(N)}$ defined in Eq.~\eqref{eq:strategy-norm} has the interpretation of an optimal distinguishing bias between processes. The cases $N=0$ and $N=1$ correspond to special cases of the strategy norm, which generalizes the trace norm for states and the diamond norm for channels to multi-round quantum strategies~\cite{gutoski2012measure,watrous2018theory}.

For $N=0$, the process has no intervening slots and testers reduce to a final measurement with elements $0\preceq T\preceq \mathds{1}$. Hence
\begin{align}
    d_{\mathrm{str}}^{(0)}(\rho,\sigma)
=
\max_{0\preceq T\preceq \mathds{1}}
\left|\tr[(\rho-\sigma)T]\right|
=
\tfrac12\|\rho-\sigma\|_1
\end{align}
which is precisely the trace distance~\cite{Nielsen_Chuang_2010}.
For $N=1$, a tester consists of an initial state preparation, possibly entangled with an ancillary memory, followed by a final measurement after the channel. This is precisely how one would evaluate the diamond-norm distance~\cite{Kitaev1997QuantumComputations,Watrous2009SDP}, and therefore
\begin{align}
d_{\mathrm{str}}^{(1)}(\mathcal{P},\mathcal{Q})
=
\tfrac12\|\mathcal{P}-\mathcal{Q}\|_\diamond. 
\end{align}

\section{Strategy-norm distinguishability and SDP formulation}
\label{app:strategy-norm-sdp}

The (finite-time) strategy-norm distance is defined as the optimal distinguishability of two $N$-step processes over all admissible testers \cite{gutoski2012measure,watrous2018theory}. Writing $\Delta^{N}:=\mathbf{P}^{N}-\mathbf{Q}^{N}$, one may express the strategy-norm distinguishability as
\begin{align}\label{eq:strategy_norm_app}
\begin{split}
    d_{\mathrm{str}}^{(N)}(\mathbf{P},\mathbf{Q})
&:=
\sup_{\mathbf{T}^{N}_j}
\Big|
\tr\!\left(\mathbf{P}^{N}[\mathbf{T}^{N}_j]-\mathbf{Q}^{N}[\mathbf{T}^{N}_j]\right)
\Big|
\\&=
\sup_{\Upsilon^{\mathbf{T}^{N}_j}\succeq 0}
\Big|
\tr\!\left(\Upsilon^{\Delta^{N}}\,\Upsilon^{\mathbf{T}^{N}_j}\right)
\Big|,
\end{split}
\end{align}
where the optimization is over tester elements satisfying the constraints in
Eqs.~\eqref{eqn:tester_constraints_app} and \eqref{eq:tester_element_constraints_app}. For $N=0$ this reduces to trace distance between states; for $N=1$ it reduces to the diamond-norm distance between channels ~\cite{gutoski2012measure,watrous2018theory}. The strategy-norm distance admits an SDP characterisation \cite{Watrous2009SDP,gutoski2012measure}. One convenient form is the following,
\begin{align}\label{eqn:SDP_strategy_norm_app}
d_{\mathrm{str}}^{(N)}(\mathbf{P},\mathbf{Q})
=
\max_{\Upsilon^{\mathbf{T}^{N}_j}}
\Big|\tr(\Upsilon^{\Delta^{N}} \Upsilon^{\mathbf{T}^{N}_j})\Big|,
\end{align}
subject to the tester constraints
\begin{subequations}\label{eqn:SDP_constraints_app}
    \begin{align}
    0\leq\Upsilon^{\mathbf{T}^{k}_j}&\leq\mathds{1}
    \qquad \forall\,k,j,
    \\
    \sum_j\Upsilon^{\mathbf{T}^{N}_j}&=\Upsilon^{\mathbf{T}^{N}},
    \\
    \tr_{\mathcal{H}_{S_{k+1}}}(\Upsilon^{\mathbf{T}^{k}})
    &=
    \Upsilon^{\mathbf{T}^{k-1}}\otimes\mathds{1}_{\mathcal{H}_{S'_k}}
    \qquad \forall\,0\leq k\leq N.
    \end{align}
\end{subequations}

\section{$\mathsf{MAD}$ testers simulate any finite-time testers}
\label{app:mad-simulates-tester}

This appendix proves Thm.~\ref{thm:mad-compilation}. We show that any finite-time tester can be implemented by probing with a single recurrent device whose internal registers include a coherent memory, a classical register which stores outcomes, and a classical counter that directs where the outcomes are stored and routes time-dependent instruments, if required.

From Lemma~\eqref{lem:Realisation of testers}, there exist auxiliary memory Hilbert spaces $\mathcal{H}_{M_k}\equiv {M_k}$ and CPTP maps
\begin{align}
    \mathcal{T}^{(k:k+1)} :
 ~&S_k'M_k\mapsto S_{k+2}M_{k+2},
\end{align}
together with a final POVM $\{E_j\}_j$ on $\mathcal{H}_{\mathrm{m}_N}$, such that when this sequential device is wired into any process $\mathbf{P}^{N}$, the resulting outcome statistics coincide with those of the tester $\{\mathbf{T}^{N}_j\}_j$. Since $N$ is finite, only finitely many intermediate memory spaces arise. Hence there exists a Hilbert space $\mathcal{H}_{\mathrm{mem}}$ of sufficient memory where
\[
d_{\mathrm{mem}}:=\max_{0\le k\le N}\dim(\mathcal{H}_{\mathrm{M_k}}).
\]
Accordingly, without loss of generality, all step maps may be regarded as acting on a common memory space. In particular, for fixed finite time $N$, one may always choose $d_{\mathrm{suf}}<\infty$. 
We now introduce a counter register with orthonormal basis $\{\ket{0}_t,\ldots,\ket{N}_t\}$ and a coherent memory space with dimension of $d_A$. The counter stores the current time step and allows the same physical device to apply the appropriate map at each round, or simply to point to the correct classical register to store the outcome for that time step. The classical memory is initialized in $\dyad{0}^{\otimes N}_C$ and the counter at $\dyad{0}_t$. Then the $\mathsf{MAD}$ $\mathsf{T}$ acts on the system and quantum ancillary spaces at $k^{\text{th}}$ time step as
\begin{widetext}
\begin{align}
\begin{split}
        &\mathsf{T}\otimes\mathds{1}_E\left[\sum_{x_{0:k-1}}\hat{\rho}_{ESA}^{\mathsf{T},x_{0:k-1}}\otimes\dyad{k}_t\otimes\dyad{x_0,x_1,\!...,x_{k-1},0,\!...,0}_C\right]
    \\&=\sum_{x_{0:k-1}}\sum_{x_k}\mathcal{T}^{(x_k)}\otimes\mathds{1}_E[\hat{\rho}_{ESA}^{\mathsf{T},x_{0:k-1}}]\otimes\dyad{k+1}_t\otimes{\dyad{x_0,x_1,\!...,x_k,0,\!...,0}_C}
    \\&=\sum_{x_{0:k}}\hat{\rho}_{ESA}^{\mathsf{T},x_{0:k}}\otimes\dyad{k+1}_t\otimes{\dyad{x_0,x_1,\!...,x_k,0,\!...,0}_C}
\end{split}
\end{align}
\end{widetext}
such that the $k^{\text{th}}$ classical register is updated with the output of the tester at that time step. Since the system may be entangled with a larger environment, we note that the $\mathsf{MAD}$ acts together with an identity on the environment $\mathds{1}_E$ on the full environment, system and ancillary spaces.

The state $\hat{\rho}^{\mathsf{T},x_{0:k}}_{ESA}$ denotes the subnormalised branch state on the environment, system, and coherent ancillary spaces, conditioned on the outcome record $x_{0:k}$ after the $k^\text{th}$ application of the tester. Thus, $N+1$ successive applications of $\mathsf{T}$ reproduce the original sequence of tester operations
\begin{align}
\mathcal{T}^{(0)},\ldots,\mathcal{T}^{(N)}
\end{align}
on the system and memory, with the counter register selecting the appropriate operation at each step. Equivalently, the counter runs through the step labels $\ket{0}_t,\ldots,\ket{N}_t$, so that the explicit time dependence of the original tester is absorbed into a classical register while the same global operation $\mathsf{T}$ is applied at every step. The coherent ancillary space must be large enough to contain the largest intermediate quantum memory space of the original tester; hence we require
\begin{align}\label{eq:ancilla-space-app}
d_A \geq d_{\mathrm{suf}}.
\end{align}
The additional counter and outcome registers are classical memory resources and are not counted in the coherent memory dimension $d_A$. By construction, when the $\mathsf{MAD}$ tester $(\mathsf{T},\rho_{A,0}^{\mathsf{T}})$ is wired into any process $\mathbf{P}^{N}$, it reproduces exactly the same final outcome statistics as the original tester $\{\mathbf{T}^{N}_j\}_j$. In particular, if $\{\mathbf{T}^{N*}_j\}_j$ is an optimal tester for distinguishing two processes over time $N$, the corresponding compiled $\mathsf{MAD}$ reproduces the same optimal binary discrimination performance, and hence achieves the strategy-norm value for that finite time. This proves Thm.~\ref{thm:mad-compilation} with coherent ancilla dimension bounded as in Eq.~\eqref{eq:ancilla-space-app}.

\section{Proof of Theorem~\ref{thm:stepwise_distinguishability_recurrent}}
\label{app:proof_of_thm_stepwise}

Consider the recurrent process $\mathbf{H}_H^N$ acting on a system $(s)$ and environment $(e)$,
\begin{align}
\mathbf{H}^{N}_{H}
:=
\tr_{E_{N+1}}\!\left(
{H_{SE}}\circ_{E_N}\dots\circ_{E_3}{H_{SE}}\circ_{E_2} {H_{SE}}\circ_{E_1}[\rho_{SE,0}^{\mathbf{H}}]
\right),
\end{align}
where $\tr_{E_{N+1}}$ is the trace over the final environment space and $\circ_{E_k}$ is the composition over the $k^\text{th}$ environment. The initial state on the system and environment spaces is given by $\rho_{SE,0}^{\mathbf{H}}$. Note that two hypotheses may differ only in the initial state while the repeating map $H$ may be the same. If the process is measured repeatedly by applying the same $\mathsf{MAD}$ $\mathsf{T}=(\{\mathcal{T}_{SA}^{(x)}\}_x,N)$ to the system and an additional quantum ancillary space, then the evolution of the environment, system, and quantum ancillary spaces together is governed by the repeated total instrument $\{\mathcal{N}_x^{\mathsf{T},H}\}_x$,
\begin{align}
    \mathcal{N}_x^{\mathsf{T},H}:=
    ({H}_{SE}\otimes \mathds{1}_A)\circ(\mathcal{T}^{(x)}_{SA}\otimes \mathds{1}_E).
\end{align}
For each fixed hypothesis $H$, the collection $\{\mathcal{N}_x^{\mathsf{T},H}\}_x$ is a valid quantum instrument, since $\sum_x \mathcal{N}_x^{\mathsf{T},H}$ is completely positive and trace preserving.

The output state on the environment, system, and ancillary spaces after a sequence of outcomes $x_{0:k}=(x_0,\dots,x_k)$ is given by
\begin{align}
\hat{\gamma}_{x_{0:k}}^{\mathsf{T},H}
:=
\mathcal{N}_{x_k}^{\mathsf{T},H}\circ \mathcal{N}_{x_{k-1}}^{\mathsf{T},H}\circ\dots\circ \mathcal{N}_{x_0}^{\mathsf{T},H}
[\rho_{SE,0}^{\mathbf{H}}\otimes\rho_{A,0}^{\mathsf{T}}].
\end{align}
In other words, $\hat{\gamma}_{x_{0:k}}^{\mathsf{T},H}$ evolves recursively as
\begin{align}
\hat{\gamma}_{x_{0:k+1}}^{H}
=
\mathcal{N}_{x_{k+1}}^{\mathsf{T},H}[\hat{\gamma}_{x_{0:k}}^{\mathsf{T},H}],
\end{align}
with
\begin{align}
\hat{\gamma}_{x_0}^{H}
:=
\mathcal{N}_{x_0}^{\mathsf{T},H}[\rho_{SE,0}^{\mathbf{H}}\otimes\rho_{A,0}^{\mathsf{T}}].
\end{align}
Moreover, we define the subnormalized outcome-dependent states on the system and quantum ancillary spaces alone, obtained by tracing out the environment,
\begin{align}
\begin{split}
    \hat{\sigma}_{x_{0:k}}^{\mathsf{T},H}
    &:=
    \tr_{E_{k+1}}(\hat{\gamma}_{x_{0:k}}^{\mathsf{T},H})
    \\
    &=
    \mathbf{H}_{H}^k[\mathbf{T}^{k}_{x_{0:k}}],
\end{split}
\end{align}
where the second line is written in terms of the process acting on the tester element associated with the outcome record $x_{0:k}$.

Since $\{\mathcal{N}_x^{\mathsf{T},H}\}_x$ is a valid quantum instrument, equivalently since $\{\mathbf{T}^{k}_{x_{0:k}}\}_{x_{0:k}}$ is a valid tester for all $k$, we have
\begin{align}
    \sum_{x_{0:k}}\tr(\hat{\gamma}_{x_{0:k}}^{\mathsf{T},H})
    =
    \sum_{x_{0:k}}\tr(\hat{\sigma}_{x_{0:k}}^{\mathsf{T},H})
    =
    \sum_{x_{0:k}}\tr(\mathbf{H}_{H}^k[\mathbf{T}^{k}_{x_{0:k}}])
    =
    1
    \qquad
    \forall~k .
\end{align}

Next, we consider the $x_{0:k}$ branchwise difference between two recurrent processes $\mathbf{P}_P$ and $\mathbf{Q}_Q$ after the $k^\text{th}$ iteration of the maps $P$ and $Q$, respectively. We express this difference in terms of the marginal states above as
\begin{align}
\begin{split}
    \Delta_{x_{0:k}}
    &:=
    \hat{\sigma}_{x_{0:k}}^{\mathsf{T},P}-\hat{\sigma}_{x_{0:k}}^{\mathsf{T},Q}
    \\
    &=
    \tr_{E_{k+1}}\!\left(
    \hat{\gamma}_{x_{0:k}}^{\mathsf{T},P}-\hat{\gamma}_{x_{0:k}}^{\mathsf{T},Q}
    \right).
\end{split}
\end{align}
At the following time step we have
\begin{align}
\Delta_{x_{0:k+1}}
=
\tr_{E_{k+2}}\!\left(\mathcal{N}^{\mathsf{T},P}_{x_{k+1}}[\hat{\gamma}^{\mathsf{T},P}_{x_{0:k}}]\right)
-
\tr_{E_{k+2}}\!\left(\mathcal{N}^{\mathsf{T},Q}_{x_{k+1}}[\hat{\gamma}^{\mathsf{T},Q}_{x_{0:k}}]\right).
\end{align}
Adding and subtracting
$\tr_{E_{k+2}}(\mathcal{N}^{\mathsf{T},P}_{x_{k+1}}[\hat{\gamma}^{\mathsf{T},Q}_{x_{0:k}}])$, we obtain
\begin{widetext}
    \begin{align}
\begin{split}
    \Delta_{x_{0:k+1}}
    =
    \,
    \tr_{E_{k+2}}\!\left(
    \mathcal{N}^{\mathsf{T},P}_{x_{k+1}}
    [\hat{\gamma}^{\mathsf{T},P}_{x_{0:k}}-\hat{\gamma}^{\mathsf{T},Q}_{x_{0:k}}]
    \right)
    +
    \tr_{E_{k+2}}\!\left(
    (\mathcal{N}^{\mathsf{T},P}_{x_{k+1}} - \mathcal{N}^{\mathsf{T},Q}_{x_{k+1}})
    [\hat{\gamma}^{\mathsf{T},Q}_{x_{0:k}}]
    \right).
\end{split}
\end{align}
\end{widetext}
Taking the trace norm and applying the triangle inequality gives
\begin{widetext}
    \begin{align}
\begin{split}
   \big|\!\big|\Delta_{x_{0:k+1}}\big|\!\big|_1
   \leq
   \,
   \big|\!\big|
   \tr_{E_{k+2}}\!\left(
   \mathcal{N}^{\mathsf{T},P}_{x_{k+1}}
   [\hat{\gamma}^{\mathsf{T},P}_{x_{0:k}}-\hat{\gamma}^{\mathsf{T},Q}_{x_{0:k}}]
   \right)
   \big|\!\big|_1
   +
   \big|\!\big|
   \tr_{E_{k+2}}\!\left(
   (\mathcal{N}^{\mathsf{T},P}_{x_{k+1}} - \mathcal{N}^{\mathsf{T},Q}_{x_{k+1}})
   [\hat{\gamma}^{\mathsf{T},Q}_{x_{0:k}}]
   \right)
   \big|\!\big|_1 .
\end{split}
\end{align}
\end{widetext}

From contractivity of quantum instruments,
\begin{align}
    \sum_{x_{k+1}}
    \big|\!\big|
    \tr_{E_{k+2}}\!\left(
    \mathcal{N}^{\mathsf{T},P}_{x_{k+1}}
    [\hat{\gamma}^{\mathsf{T},P}_{x_{0:k}}-\hat{\gamma}^{\mathsf{T},Q}_{x_{0:k}}]
    \right)
    \big|\!\big|_1
    \leq
    \big|\!\big|
    \hat{\gamma}^{\mathsf{T},P}_{x_{0:k}}-\hat{\gamma}^{\mathsf{T},Q}_{x_{0:k}}
    \big|\!\big|_1 .
\end{align}
Therefore, summing over all words
$x_{0:k+1}=(x_0,\dots,x_k,x_{k+1})\equiv(x_{0:k},x_{k+1})$, we obtain
\begin{widetext}
    \begin{align}
\sum_{x_{0:k+1}}
\big|\!\big|\Delta_{x_{0:k+1}}\big|\!\big|_1
\leq
\,
\sum_{x_{0:k}}
\big|\!\big|
\hat{\gamma}^{\mathsf{T},P}_{x_{0:k}}-\hat{\gamma}^{\mathsf{T},Q}_{x_{0:k}}
\big|\!\big|_1
+\!\!
\sum_{x_{0:k},x_{k+1}}\!\!
\big|\!\big|\!
\tr_{E_{k+2}}\!\left(
(\mathcal{N}^{\mathsf{T},P}_{x_{k+1}} - \mathcal{N}^{\mathsf{T},Q}_{x_{k+1}})
[\hat{\gamma}^{\mathsf{T},Q}_{x_{0:k}}]
\right)\!\!
\big|\!\big|_1 .
\label{eq:Deltaxk+1_inequality}
\end{align}
\end{widetext}

Finally, we note that the distinguishability of the recurrent processes with repeated map $H\in\{P,Q\}$ can be expressed in terms of the distinguishability of the states
\begin{align}
    \rho_{k,\mathsf{T}}^{H}
    =
    \sum_{x_{0:k}\in\mathcal{X}^k}
    |x_{0:k}\rangle\!\langle x_{0:k}|_{C}
    \otimes
    \hat{\sigma}^{H}_{x_{0:k}},
\end{align}
which are defined on the system, quantum ancillary space, and classical outcome register. By storing each outcome-dependent state together with a classical flag state, the state $\rho_{k,\mathsf{T}}^{H}$ specifies the process--tester interaction completely at the level accessible to the final measurement.

With this in mind, we define
\begin{align}
\delta_{\mathsf{T}}^k
:=
\frac12
\sum_{x_{0:k}}
\big|\!\big|\Delta_{x_{0:k}}\big|\!\big|_1 .
\end{align}
Since the propagation term in Eq.~\eqref{eq:Deltaxk+1_inequality} depends on the full environment--system--ancilla branch states before tracing out the environment, we also define the corresponding full branch distinguishability
\begin{align}
\tilde{\delta}_{\mathsf{T}}^k
:=
\frac12
\sum_{x_{0:k}}
\big|\!\big|
\hat{\gamma}^{\mathsf{T},P}_{x_{0:k}}-\hat{\gamma}^{\mathsf{T},Q}_{x_{0:k}}
\big|\!\big|_1 .
\end{align}
By contractivity of the partial trace,
\begin{align}
    \delta_{\mathsf{T}}^k
    \leq
    \tilde{\delta}_{\mathsf{T}}^k .
\end{align}
We also define the additional local distinguishability-generation term
\begin{align}
\epsilon_{\mathsf{T}}^k
:=
\frac12
\sum_{x_{0:k},x_{k+1}}
\big|\!\big|
\tr_{E_{k+2}}\!\left(
(\mathcal{N}^{\mathsf{T},P}_{x_{k+1}} - \mathcal{N}^{\mathsf{T},Q}_{x_{k+1}})
[\hat{\gamma}^{\mathsf{T},Q}_{x_{0:k}}]
\right)
\big|\!\big|_1 .
\end{align}
Multiplying both sides of Eq.~\eqref{eq:Deltaxk+1_inequality} by one half gives
\begin{align}
    \delta_{\mathsf{T}}^{k+1}
    \leq
    \tilde{\delta}_{\mathsf{T}}^k
    +
    \epsilon_{\mathsf{T}}^k 
\end{align}
as claimed. If no distinguishability is hidden in the inaccessible environment so that
\begin{align}
    \tilde{\delta}_{\mathsf{T}}^k
    =
    \delta_{\mathsf{T}}^k ,
\end{align}
then the bound reduces to
\begin{align}
    \delta_{\mathsf{T}}^{k+1}
    \leq
    \delta_{\mathsf{T}}^k
    +
    \epsilon_{\mathsf{T}}^k .
\end{align}

\end{document}